\documentclass[amsmath,11pt]{article}
\usepackage{times}
\usepackage{fullpage}
\usepackage{amsmath,amsfonts,amssymb,amsthm}
\usepackage{url}
\usepackage{enumerate}
\usepackage{algorithm}
\usepackage[noend]{algorithmic}

\theoremstyle{plain}            % following are "theorem" style
\newtheorem{theorem}{Theorem}[section]
\newtheorem{lemma}[theorem]{Lemma}

\newtheorem{claim}[theorem]{Claim}

\theoremstyle{definition}       % following are def style

\theoremstyle{remark}           % following are remark style

\numberwithin{equation}{section}

\DeclareMathOperator{\poly}{poly}
\DeclareMathOperator{\polylog}{polylog}

\DeclareMathOperator*{\E}{E}

\newcommand{\eqdef}{\mathbin{\stackrel{\rm def}{=}}}

\makeatletter
\def\imod#1{\allowbreak\mkern4mu({\operator@font mod}\,\,#1)}
\makeatother

\newcommand{\N}{\ensuremath{\mathbb{N}}}
\newcommand{\Q}{\ensuremath{\mathbb{Q}}}
\newcommand{\R}{\ensuremath{\mathbb{R}}}

\newcommand{\Z}{\ensuremath{\mathbb{Z}}}

\newcommand{\lat}{\mathcal{L}}

\newcommand{\veca}{\ensuremath{\mathbf{a}}}
\newcommand{\vecb}{\ensuremath{\mathbf{b}}}
\newcommand{\vecc}{\ensuremath{\mathbf{c}}}

\newcommand{\vece}{\ensuremath{\mathbf{e}}}

\newcommand{\vecs}{\ensuremath{\mathbf{s}}}
\newcommand{\vect}{\ensuremath{\mathbf{t}}}

\newcommand{\vecv}{\ensuremath{\mathbf{v}}}
\newcommand{\vecw}{\ensuremath{\mathbf{w}}}
\newcommand{\vecx}{\ensuremath{\mathbf{x}}}
\newcommand{\vecy}{\ensuremath{\mathbf{y}}}
\newcommand{\vecz}{\ensuremath{\mathbf{z}}}
\newcommand{\veczero}{\ensuremath{\mathbf{0}}}
\newcommand{\parl}{\mathcal{P}}

\newcommand{\pr}[2]{\langle{#1, #2}\rangle}

\newcommand{\set}[1]{\{{#1}\}}

\newcommand{\floor}[1]{\lfloor{#1}\rfloor}

\newcommand{\ceil}[1]{\lceil{#1}\rceil}

\def\eps{\epsilon}

\def\linsp{\mathrm{span}}

\def\vol{\mathrm{vol}}

\def\minuszero{\setminus \{\veczero\}}

\DeclareMathOperator*{\argmax}{arg\,max}

\title{A Deterministic Polynomial Space Construction for $\eps$-nets under any Norm}

\author{Daniel Dadush\thanks{Department of Computer Science, New York University. {\tt Email: \hspace{-1em}dadush@cs.nyu.edu}.}}

\begin{document}

\maketitle

\begin{abstract}
We give a deterministic polynomial space construction for nearly optimal $\eps$-nets with respect to any input $n$-dimensional convex body $K$ and
norm $\|\cdot\|$. More precisely, our algorithm can build and iterate over an $\eps$-net of $K$ with respect to $\|\cdot\|$ in time $2^{O(n)} \times
(\text{ size of the optimal net })$ using only $\poly(n)$-space. This improves on previous constructions of~\cite{conf/stoc/ASLV13} which achieve
either a $2^{O(n)}$ approximation or an $n^{O(n)}$ approximation of the optimal net size using $2^n$ space and $\poly(n)$-space
respectively. As in~\cite{conf/stoc/ASLV13}, our algorithm relies on the mathematically classical approach of
building thin lattice coverings of space, which reduces the task of constructing $\eps$-nets to the problem of enumerating lattice points. Our main
technical contribution is a deterministic $2^{O(n)}$-time and $\poly(n)$-space construction of thin lattice coverings of space with respect to any
convex body, where enumeration in these lattices can be efficiently performed using $\poly(n)$-space. This also yields the first \emph{existential}
construction of $\poly(n)$-space enumerable thin covering lattices for general convex bodies, which we believe is of independent interest. Our
construction combines the use of the M-ellipsoid from convex geometry~\cite{M86} with lattice sparsification and
densification techniques~\cite{Rogers1950,conf/soda/cvp/DK13}.

As an application, we give a $2^{O(n)}(1+1/\eps)^n$ time and $\poly(n)$-space deterministic algorithm for computing a $(1+\eps)^n$ approximation to
the volume of a general convex body, which nearly matches the lower bounds for volume estimation in the oracle model (the dependence on $\eps$ is
larger by a factor $2$ in the exponent). This improves on the previous results of~\cite{journal/pnas/DV13}, which gave the above result only for 
symmetric bodies and achieved a dependence on $\eps$ of $(1+\log^{5/2}(1/\eps)/\eps^2)^n$. 
\end{abstract}

% \footnote{It was communicated to the author~\cite{error/V13} that the Las Vegas construction achieving a $2^{O(n)}$ approximation of the
% optimal net size using $\poly(n)$-space claimed in~\cite{conf/stoc/ASLV13} is flawed. In particular, it only achieves an $n^{O(n)}$ approximation. A
% revision of the paper stating this is pending.}
% \thispagestyle{empty}

\newpage
\setcounter{page}{1}

\section{Basic Concepts}
\label{sec:basic}
\paragraph{{\bf Convexity.}} Define $B_2^n = \set{\vecx: \|\vecx\|_2 \leq 1}$ to be the unit Euclidean ball in $\R^n$.  For sets $A,B \subseteq
\R^n$, $s,t \in \R$, we define the Minkowski sum $sA+tB = \set{s\veca + t\vecb: \veca \in A, \vecb \in B}$.  A convex body $K \subseteq \R^n$ is a
compact convex set with non-empty interior. For any convex set $K$, we have the algebra $sK+tK = (s+t)K$ for $s,t \geq 0$. $K$ is symmetric if $K =
-K$ and $\veczero$-centered if $\veczero$ is in the interior of $K$. For a $\veczero$-centered convex body, we define the polar $K^\circ = \set{\vecx
\in \R^n: \pr{\vecx}{\vecy} \leq 1 ~\forall \vecy \in K}$. We let $\|\vecx\|_K = \inf \set{s \geq 0: \vecx \in sK}$ denote the \emph{gauge function}
of $K$. Here $\|\cdot\|_K$ satisfies all norm properties except symmetry when $K$ is $\veczero$-centered and induces a norm in the
usual sense when $K$ is symmetric.

For two sets $A,B \subseteq \R^n$, we denote the \emph{covering number} of $A$ with respect to $B$ is
\[
N(A,B) = \min \set{|T|: T \subseteq \R^n, A \subseteq T + B} 
\]
$A,B$ have covering numbers bounded by $(c_1,c_2)$ if $N(A,B) \leq c_1$ and $N(B,A) \leq c_2$.

We define the ellipsoid $E(A) = \set{\vecx \in \R^n: \vecx^T A \vecx \leq 1}$, where $A$ is an $n \times n$ symmetric positive definite matrix. From here
one has that $E(A) = A^{-1/2} B_2^n$ and $\vol_n(E(A)) = \det(A)^{-1/2} \vol_n(B_2^n)$.

For an $n$ dimensional convex body $K$, we say that an ellipsoid $E$ is an $M$-ellipsoid of $K$ if $K,E$ have covering numbers bounded by $2^{O(n)}$
(see Section~\ref{sec:prelims-c} for more details).

\paragraph{{\bf Computational Model:}} When interacting algorithmically with convex bodies, we will assume that they are presented by membership
oracles in the standard way (see Section~\ref{sec:prelims-c} for more details). The complexity of our algorithms will be measured by the number of arithmetic operations and oracle calls.

\paragraph{{\bf Lattices.}} An $n$-dimensional lattice $\lat \subseteq \R^n$ is the integer span of a basis $B=(\vecb_1,\dots,\vecb_n)$ of $\R^n$. We
also use the notation $\lat(B)$ to denote the lattice spanned by a basis $B$. The determinant $\det(\lat)$ of $\lat$ is defined as $|\det(B)|$.
We define the (symmetric) parallelepiped with respect to $B$ as $\parl(B) = B[-1/2,1/2]^n$.  Let $M \subseteq \lat$ be a sublattice of $\lat$. We
define the quotient group $\lat \imod{M} = \set{M + \vecy: \vecy \in \lat}$, i.e.~the cosets of $M$ with respect to $\lat$. Let $[\lat : M]$ denote
the index of $\lat$ with respect to $M$, where $[\lat : M] = |\lat \imod{M}|$. If $[\lat : M] < \infty$, then $[\lat : M] = \det(M)/\det(\lat)$ and
$\dim(M) = \dim(\lat)$. For $p \in \N$ the group $\lat/p \imod{\lat} = \set{\lat + B\veca/p: \veca \in \set{0,\dots,p-1}^n}$, where group addition
corresponds to adding the coefficient vectors modulo $p$, and hence, $\lat/p \imod{\lat} \cong \Z_p^n$.

Let $K$ be a $\veczero$-centered convex body. We denote distance between a point $\vecx \in \R^n$ and $\lat$ under $\|\cdot\|_K$ as $d_K(\lat,\vecx) =
\min_{\vecy \in \lat} \|\vecy-\vecx\|_K$. The covering radius of $K$ with respect to $\lat$ is 
\[
\mu(K,\lat) = \inf \set{s \geq 0: \lat+sK = \R^n} = \max_{\vecx \in \R^n} d_K(\lat,\vecx) \text{.}
\]
$\lat$ is $K$-covering if $\mu(K,\lat) \leq 1$ and $\alpha$-thin if $\vol_n(K)/\det(\lat) \leq \alpha$. We note that the notion of covering radius
makes sense for any convex body (since it can be stated independent of centering).

Let $K$ be a symmetry convex body. We define the minimum distance of $\lat$ with respect to $K$ as $\lambda_1(K,\lat) = \inf_{\vecy \in \lat
\minuszero} \|\vecy\|_K$. Let $\lambda = \lambda_1(K,\lat)$, $\mu = \mu(K,\lat)$. $\lat$ is $K$-packing if $\lambda_1(K,\lat) \geq 2$. The
\emph{packing density} of $\lat$ with respect to $K$ is $\vol_n(\lambda/2 K)/\det(\lat)$.  Note that the packing density is always less than $1$ since
the lattice shifts of $(\lambda/2) K$ are all interior disjoint. The \emph{packing to covering ratio} of $\lat$ with respect to $K$ is
$\lambda/(2\mu)$. Note if $s < \lambda/2$, i.e.~below the packing radius, then lattice shifts of $sK$ must leave parts of space uncovered. From this,
we see that the packing to covering ratio is also always less than $1$. 

Let $K$ be a $\veczero$-centered convex body. The \emph{Shortest Vector Problem (SVP)} with respect to $\lat$ and $K$ is to find a shortest non-zero
vector in $\lat$ under $\|\cdot\|_K$. The \emph{Closest Vector Problem (CVP)} with respect to $\lat$, $K$ and target $\vecx \in \R^n$ is to find a
closest lattice vector $\vecy \in \lat$ to $\vecx$ under $\|\cdot\|_K$, i.e.~that minimizes $\|\vecy-\vecx\|_K$. 

\section{Introduction}
\label{sec:intro}

The usefulness of $\eps$-nets within Computer Science for designing approximation algorithms, derandomization and in many other contexts,
is well-established. In this paper, we will explore algorithms for constructing such nets in a general geometric setting. In particular, we will be
interested in the following algorithmic task: given $n$-dimensional convex bodies $C$ and $K$, construct a covering of $C$ by $K$ in time at most $f(n) N(C,K)$. 
In this language, constructing an $\eps$-net for $C$ under a given norm $\|\cdot\|$, corresponds to the covering problem where $K = \set{\vecx \in
\R^n: \|\vecx\| \leq \eps}$. 

In this general context, the problem of algorithmically constructing such coverings was first studied in~\cite{conf/stoc/ASLV13}. Here, they showed
that such coverings can be used to yield an additive PTAS for $2$-player Nash equilibria and the Densest Subgraph problem, in the case where the sum
of the payoff matrices (for $2$-player Nash) or the adjacency matrix (for Densest Subgraph) has logarithmic $\eps$-rank. To build the coverings, they
relied on constructions of thin lattice coverings of space. Based on these they gave two deterministic constructions for $\eps$-nets in the case where
the covering body is symmetric: one achieving $f(n)=n^{O(n)}$ using $\poly(n)$ space based on good ellipsoidal roundings, and another achieving $f(n)
= 2^{O(n)}$ using $2^n$ space based on a construction of Rogers~\cite{Rogers1950}\footnote{The original paper also claimed a Las Vegas construction
achieving $f(n)=2^{O(n)}$ using $\poly(n)$-space, but this construction was flawed~\cite{error/V13}. However, this does not affect the time complexity
of their main algorithm and its applications.}. 

Our goal will be to build such coverings for general convex bodies with $f(n) = 2^{O(n)}$ using only $\poly(n)$-space. We will also give applications
to the problems of deterministic volume estimation for convex bodies, estimating norms of linear operators, constructing polyhedral approximations,
and of derandomizing certain lattice algorithms. 

We note that exponential approximation factors are very natural in the geometric setting, since even covering an $n$-dimensional convex body $C$ by a $(1/2)C$
must have size at least $2^n$ just by comparing volumes.  Hence, even small perturbations of either of the bodies $C$ and $K$ will generically change
the covering numbers by an exponential factor. 

\subsection{Thin Lattice Coverings}
\label{sec:intro-thin-lat}

As in~\cite{conf/stoc/ASLV13}, our method for building coverings relies on the mathematically classical approach, primarily developed by C.A. Rogers,
of building a thin lattice covering of space with respect to the covering body $K$, and restricting the covering to the body $C$ (see for
example~\cite{Rogers1950,Rogers1958,Rogers1959,RogersZong97}). 
\vspace{1em}

\noindent More formally, for any convex body $K$, we seek to build a lattice $\Lambda$ satisfying
\begin{enumerate}
\item {\bf Covering:} $\Lambda$ is $K$-covering.
\item {\bf Thinness:} $\vol_n(K)/\det(\Lambda) \leq t(n)$.
\item {\bf Enumeration Compexity:} For any convex body $C$, the lattice points in 
\[
(C-K) \cap \lat \quad\quad \left(\text{corresponding to the lattice shifts of $K$ touching $C$}\right)
\]
can be enumerated in time $f(n) N(C,K)$ using at most most $p(n)$ space.
\end{enumerate}

Our main result is as follows:\vspace{1em}

\begin{theorem}[Thin Lattice] There is a deterministic $2^{O(n)}$ time and $\poly(n)$ space construction for covering lattices with
respect to
\begin{enumerate}
\item {\bf symmetric convex bodies} satisfying $t(n)=3^n$, $f(n)=2^{O(n)}$ and $p(n)=\poly(n)$, where the constructed lattices have packing to covering ratio at least $1/3$.
\item {\bf general convex bodies} satisfying $t(n)=7^n$, $f(n)=2^{O(n)}$ and $p(n)=\poly(n)$.
\end{enumerate}
Furthermore, enumeration within these lattices can be implemented using standard Schnorr-Euchner enumeration.
\label{thm:thin-lat-informal}
\end{theorem}

Schnorr-Euchner enumeration is a basis centric form of lattice point enumeration, which uses a search tree over the basis coefficients to find lattice
points, and is perhaps the most commonly used lattice point enumeration method in practice (see Section~\ref{sec:schnorr-euchner} for more details).
We note that Theorem~\ref{thm:thin-lat-informal} gives the first \emph{existential} construction of low-space enumerable covering lattices (via
Schnorr-Euchner or any other known low-space method). Indeed, for the main class of lattices used to show the existence of thin coverings, that is the
so-called Haar (or random) lattices (see~\cite{Rogers1958,Rogers1959}), it is known that Schnorr-Euchner enumeration (and all other known low-space
enumeration methods) is in general not efficient. In particular, for Haar lattices it can shown that the Schnorr-Euchner enumeration complexity for a
scaled Euclidean ball can be an $n^{\Omega(n)}$ factor larger than the number of points in the ball (see for example, section $2$ in
\cite{eprint/svp/BGJ13}). We note that these types of lattices form a main class of ``hard'' test instances for solving the classical Shortest (SVP)
and Closest Vector Problems (CVP).

As an added bonus of our construction, the covering lattices in~\ref{thm:thin-lat-informal} have a packing to covering ratio of at least $1/3$ for
symmetric bodies, and have the property that CVP under the norm for which they were constructed can be solved in $2^{O(n)}$ time and $\poly(n)$ space
(since this reduces to enumeration within the covering body). We note that while building thin covering lattices for $\ell_p$ norms is trivial --
$2n^{-1/p} \Z_n$ is a $2^{O(n)}$-thin covering lattice for the $\ell_p$ norm -- building ones with constant packing to covering ratio is not. In fact,
even for the $\ell_2$ norm, there is no known explicit construction of such a lattice. Furthermore, as mentioned above, the main probabilistic
constructions do not currently have space efficient CVP solvers. While these properties are not directly used in our applications, we
believe they might be useful elsewhere, such as in lattice based schemes for Locality Sensitive Hashing (see~\cite{conf/focs/IndykA06} for an
application using the $24$-dimensional Leech lattice).

From the perspective of space usage, if one is willing to use exponential space, then lattice point enumeration can be performed efficiently within any
thin covering lattice. In particular, given a $t(n)$-thin covering lattice, the M-ellipsoid covering and Voronoi cell based enumeration algorithm
of~\cite{journal/siamjc/MV13,conf/focs/svp/DPV11,thesis/D12} can be used to achieve enumeration complexity $f(n) = 2^{O(n)} t(n)$. Given this, the main contribution in
Theorem~\ref{thm:thin-lat-informal} is in building thin covering lattices with low enumeration complexity. We remark that in the context of lattice
problems such as CVP and SVP, there is a general dichotomy in the known time / space tradeoffs, where on one end we have
$\poly(n)$ space and $n^{O(n)}$ time algorithms and on the other $2^{O(n)}$ time and space algorithms. From this perspective,
Theorem~\ref{thm:thin-lat-informal} yields a non-trivial example where this type of dichomotomy is in fact unnecessary.  

%Our result is based on making a classical ``greedy'' construction of Rogers~\cite{Rogers1950} more space efficient. 

\paragraph{{\bf From Thinness to Coverings.}} Following \cite{RogersZong97}, we outline how to use a thin $K$-covering lattice $\Lambda$ to
recover a nearly optimal covering of any convex body $C$ by $K$. Firstly, we note that the $K$-covering property $\Lambda + K = \R^n$ implies that all
the lattice shifts of $K$ touching $C$ must form a covering of $C$. Formally, the set
\[
T = \set{\vecy \in \Lambda: (\vecy + K) \cap C \neq \emptyset} = \set{\vecy \in \Lambda: \vecy \in C-K} = (C-K)  \cap \Lambda \text{,}
\]
satifies $C \subseteq T+K$. Furthermore, the same argument works for any shift $\Lambda + \vecx$ as well, in that $T =  (C-K) \cap (\Lambda + \vecx)$
also forms yields a covering of $C$ by $K$. If one estimates the number of lattice points in $C-K$ using the so-called Gaussian heuristic we expect that 
\[
|(C-K) \cap \Lambda| \approx \vol_n(C-K)/\det(\Lambda) \text{.}
\]
In fact, a straighforward averaging argument reveals that if we pick a uniform coset $\vecx \leftarrow \R^n \imod{\Lambda}$ 
\[
\E_{\vecx}[(C-K) \cap (\Lambda+\vecx)] = \vol_n(C-K)/\det(\Lambda) \text{.}
\]
Hence a covering of this size always exists via the probabilistic method. However, since our goal is to get deterministic algorithms, we will show
later that the central coset, i.e.~$\Lambda + \vecx = \Lambda$, yields a covering of size at worst a $2^{O(n)}$ factor larger than the Gaussian
heuristic in this setting. For the time being, let us therefore assume that the central coset achieves this bound. Now let $T$ denote an optimal
covering of $C$ by $K$. From here, we see that
\[
\vol_n(C-K) \leq \vol_n(T+K-K) \leq \vol_n(K-K) |T|  \leq 4^n \vol_n(K) N(C,K)\text{,} 
\]  
where the last step follows by the Rogers-Shepard inequality~\cite{RS57}
\[
\vol_n(K-K) \leq \binom{2n}{n} \vol_n(K)
\]
(if $K$ is symmetric we get $\vol_n(K-K) = 2^n \vol_n(K)$). Putting it all together, we get that 
\begin{equation}
\label{eq:lat-cover-bound}
|(C-K) \cap \Lambda| \leq 4^n (\vol_n(K)/\det(\Lambda)) N(C,K) = 4^n t(n) N(C,K) \text{.} 
\end{equation}
(the right hand side drops by $2^n$ if $K$ symmetric) where $t(n)$ is the thinness of the covering. From the above argument, we see that the thinness
of the covering lattice essentially controls the quality of the coverings we can expect to derive from it.  

From the perspective of optimizing thinness, as mentioned previously, it has long been known that Haar lattices yield extremely efficient coverings.
In particular, in a rather surprising result, Rogers~\cite{Rogers1959} showed that Haar lattices modified by a small number of additional ``random
densification'' steps can be used to construct covering lattices of thinness $n^{\log \log n + O(1)}$ for any convex body and thinness $n
\log^{O(1)}(n)$ for the Euclidean ball. This result was further extended by Butler~\cite{Butler1972}, who showed that with the same thinness one can
build lattices with packing to covering ratio at least $1/2-o(1)$ (which is conjectured to be optimal as $n \rightarrow \infty$ for $K=B_2^n$). We
note however, that even for covering lattices of thinness $1$, the above bounds on the size of the coverings is still an exponential factor larger
than $N(C,K)$. In particular, if $C = K = [0,1]^n$ and $\Lambda = \Z^n$, then the set 
\[
|(C-K) \cap (\Lambda + \vect)| = |[-1,1]^n \cap (\Z^n + \vect)| \geq 2^n, \quad \text{ for all } \vect \in \R^n \text{,}
\]
even though $N(C,K) = 1$. Hence the bound from Equation~\eqref{eq:lat-cover-bound} can be tight. Furthermore, by applying a random shift to either $C$ or $K$
in the above example, we see that approximating $N(C,K)$ to within less than $2^n$ is essentially impossible under the oracle model (since the set
cube shifts $\vect$ such that $\vect+[0,1]^n$ contains more than one vertex of $[0,1]^n$ has measure zero).

We remark that while the construction described above relies on lattice coverings of space, it can easily be adapted to use non-lattice coverings as
well (which can achieve thinness $O(n \log n)$, i.e.~somewhat better than lattice coverings). However, one main drawback of non-lattice coverings --
which usually consist of ``random'' shifts of a base lattice covering -- is that there is no generic way to certify them (i.e.~that they indeed yield
a covering). In the lattice setting, this task much simpler, since it reduces to approximating the covering radius. Here, it was shown
in~\cite{journal/cc/GuMiRe05} that a bound $\mu \leq \mu(K,\Lambda) \leq p/(p-1) \mu$ can be computed by solving $p^n$ CVPs (see Lemma
\ref{lem:cov-approx}). We note that even in the case of lattice coverings, using the above procedure, $2^{O(n)}$ time only allows us to compute a
$(1+\eps)$ approximation of the covering radius, for a fixed $\eps > 0$, and hence it is unclear how one could effectively certify thinness below
$(1+\eps)^n$.

\paragraph{{\bf Thin Lattice Constructions.}} We now discuss how one can construct thin covering lattices, and explain how our construction differs from
previous work. We will restrict here to the case where $K$ is symmetric. In the next section, we will show how to reduce the general case to the
symmetric one. 

To build intuition, we describe the first basic construction of~\cite{conf/stoc/ASLV13}. Given the initial $n$ dimensional covering body $K$, a first
natural way to get a handle on the coarse geometry of $K$ is to compute an appropriate ellipsoidal approximation. As a first try, we may attempt to
compute a good sandwiching ellipsoid $E$ for $K$, i.e.~an ellipsoid satisfying $E \subseteq K \subseteq c E$, where $c$ is small as possible. For
$n$-dimensional symmetric convex bodies sandwiching ellipsoids always exist for $c = \sqrt{n}$ (e.g.~one may use the maximum volume contained
ellipsoid), and this is tight (e.g.~the cube vs the ball). By a linear transformation -- note that all the desired properties of the covering lattice are
preserved by a simultaneous linear transformation of the lattice and covering body -- we may assume that $B_2^n \subseteq K \subseteq \sqrt{n} B_2^n$.
A simple choice of $K$-covering lattice is now $\Lambda = \frac{2}{\sqrt{n}}\Z^n$. The covering property follows from the fact that $K$
contains the cube $[\frac{-1}{\sqrt{n}},\frac{1}{\sqrt{n}}]^n \subseteq B_2^n \subseteq K$, which is the (symmetric) fundamental parallelepiped with
respect to the basis $B =(\frac{2}{\sqrt{n}} \vece_1,\dots, \frac{2}{\sqrt{n}} \vece_n)$. A first question is how thin is this lattice covering? From
the sandwiching bounds we get 
\[
\frac{\vol_n(K)}{\det(\Lambda)} \leq \frac{\vol_n(\sqrt{n} B_2^n)}{\det(\frac{2}{\sqrt{n}}\Z^n)} = (n/2)^n \vol_n(B_2^n) = 2^{\Theta(n)} n^{n/2} \text{.}
\]
Another question is how easy is enumeration in this lattice? As we will see, the main consideration will be the enumeration complexity for the
covering body $K$ itself, as this will essentially determine the $f(n)$ parameter. For our choice of lattice $\Lambda = \frac{2}{\sqrt{n}}\Z^n$, one
can consider the graph over $\Lambda$ whereby two lattice points are adjacent if their associated parallelepipeds $\parl(B) =
[\frac{-1}{\sqrt{n}},\frac{1}{\sqrt{n}}]^n$ intersect in a facet, or put more simply if their difference is in $\pm \set{\frac{2}{\sqrt{n}}
\vece_1,\dots, \frac{2}{\sqrt{n}} \vece_n}$. Here it is not hard to check that the restriction of this graph to the lattice points forming a
$\parl(B)$-tiling of $K$, that is $(K-\parl(B)) \cap \Lambda$, is connected. Furthermore, given that $\parl(B) \subseteq K$, via similar arguments to
those above the tiling has size bounded by $2^n \vol(K)/\det(\Lambda)$. Hence the points in $K \cap \Lambda$ can be enumerated by computing the
connected component of $\veczero$ in the tiling graph in $\poly(n) 2^n \vol(K)/\det(\Lambda)$ time via a depth first or breadth first search. To make
this enumeration space efficient (avoiding a linear dependence on the size of the graph), a simple line following argument shows that the edges of the
shortest path tree directed towards $\veczero$ can be computed locally. From here one can show that a traversal of the vertices of this implicit
shortest path tree can be computed in space logarithmic in the size of the graph -- which is $\poly(n)$ in this setting -- starting from $\veczero$
(see~\cite{thesis/D12} for a full exposition).  

The above construction of~\cite{conf/stoc/ASLV13} yields a $2^{O(n)} n^n$-thin $K$-covering lattice $\Lambda$ that is $\poly(n)$-space enumerable.
While this is not good enough for our purposes, we will make use of the main fact enabling low space enumeration. In particular, if a convex body $C$
has a tiling with respect to a basis parallelepiped $\parl(B)$ of size $f(n) |C \cap \Lambda|$, then the points $C \cap \Lambda$ can be enumerated in
$\poly(n)$ space and $f(n) |C \cap \Lambda|$ time. We will strengthen this observation, by showing that Schnorr-Euchner (SE) enumeration -- which always
operates using $\poly(n)$ space -- over $C \cap \Lambda$ using basis $B$ has complexity bounded by $\poly(n) N(C,\parl(B))$ (see
Lemmas~\ref{lem:se-enum} and~\ref{lem:se-cover-bounds}). Note that by definition, the parallelepiped covering number is always bounded by the size of
a parallelepiped tiling. Apart from yielding a somewhat simpler enumeration algorithm, SE enumeration will be very useful in that it will
make it easy to quantify how the enumeration complexity changes when taking sublattices or superlattices of any base lattice.  In particular, we 
show that the SE enumeration complexity for a convex body does not increase when taking sublattices, and increases by at most the index
when taking superlattices (see Lemma \ref{lem:se-robust}). 

To improve on the above construction, we will make use of three additional ingredients. Firstly, we construct a lattice basis $B$ whose parallelepiped
$\parl(B)$ has covering numbers at most $2^{O(n)}$ with respect to $K$. This can be achieved by choosing $\parl(B)$ to be a maximum volume in
inscribed parallelepiped for an $M$ Ellipsoid $E$ of $K$. We note that the ``M-Lattice'' $\lat = \lat(B)$ is used in~\cite{conf/focs/svp/DPV11} to
compute the $M$-ellipsoid covering for the lattice point enumeration algorithm. By asking for more than the sandwiching bounds achieved in the
previous construction, we get good bounds on the volume of $K$, i.e. $\det(\lat) = 2^{\Theta(n)} \vol_n(K)$ (avoiding the previous $n^n$ factor), and
- as mentioned above - we get that Schnorr-Euchner enumeration in $K$ with respect to $B$ takes at most $2^{O(n)}$ time. At this point, from the
  robustness of SE enumeration, we can reduce the covering lattice problem to building a $K$-covering lattice $\Lambda$ that is ``not too far'' from
the base lattice $\lat$. In particular, it will suffice for us if $\Lambda$ can be obtained by a sequence of sublattice and superlattice operations
over $\lat$ where the product of the indexes is at most $2^{O(n)}$ (in fact, it will be a superlattice of a sublattice).  

The remaining two ingredients are the use of lattice sparsification and densification. Here the idea will be to use sparsification to choose a
sublattice of small index which gets rid of all short lattice vectors, and to use densification to construct a superlattice of small index which
reduces the covering radius to a constant multiple of the minimum distance. 

The original construction of Rogers~\cite{Rogers1950}, which is implemented in~\cite{conf/stoc/ASLV13}, uses a ``greedy'' deterministic densification
procedure to construct a lattice with packing to covering ratio at least $1/3$. More precisely, starting from a base lattice $\lat$, Rogers looks for
a point $\vecy \in \lat/3$ that is at distance at least $\lambda_1 \eqdef \lambda_1(K,\lat)$ from $\lat$ under $\|\cdot\|_K$. If such a point $\vecy$
exists, we adjoin $\vecy$ to $\lat$ and repeat. The distance lower bound here guarantees that the minimum distance does not decrease when we adjoin
$\vecy$. Furthermore, the determinant decreases by a factor of $3$ after adjoining $\vecy$, and hence the packing density of the new lattice increases
by a factor $3$. If no such point exists, then every point in $\lat/3$ is at distance at most $\lambda_1$ from $\lat$, which implies (see Lemma
\ref{lem:cov-approx}) that $\mu(K,\lat) \leq (3/2)\lambda_1$ (i.e.~packing to covering ratio $1/3$). We note that without the symmetry assumption on
$K$, it is unclear how to derive the bound on the covering radius once the procedure terminates. A nice feature of this construction is that it can be
implemented as long as one can efficiently enumerate lattice points in the current lattice with respect to shifts of $\lambda_1 K$, where $\lambda_1$
stays fixed throughout the construction. 

When starting from an M-Lattice $\lat$ with basis $B$ (where $\parl(B)$ is fundamental parallelepiped built from an M-Ellipsoid of $K$), the
enumeration within $\lambda_1K$ can initially be done in $2^{O(n)}$ time using $\poly(n)$ space via SE enumeration, where here $\lambda_1 = O(1)$
since $\vol_n(K) \geq 2^{-O(n)} \det(\lat)$. However, the efficiency of enumeration degrades over the course of the construction as the lattice gets
denser. In particular, the enumeration complexity can jump by a $3^k$ factor after $k$ iterations, since this is the index with respect to the base
lattice. We note that the number of lattice points in any shift of $\lambda_1K$ is never larger than $5^n$ by a standard packing bound. While this
does not bound the SE enumeration complexity, it is sufficient to bound the time complexity of the M-ellipsoid and Voronoi cell based enumeration
algorithm of~\cite{journal/siamjc/MV13,conf/focs/svp/DPV11} by $2^{O(n)}$ while using $2^{O(n)}$ space. The latter method describes the implementation
in~\cite{conf/stoc/ASLV13}. Since we seek to avoid the use of exponential space, we will show how to keep SE enumeration efficient throughout the
entire procedure. Given the above reasoning, for SE enumeration to remain $2^{O(n)}$ time, one needs to ensure that the Rogers densification procedure
terminates in $O(n)$ steps.

The only general bound on the iteration complexity of Rogers densification procedure is based on the packing density of the base lattice, i.e.
$\vol_n((\lambda_1/2)K)/\det(\lat)$. If the base lattice has packing density $3^{-l}$, then since the packing density increases by a factor $3$ at
each iteration, the number of iterations must be bounded by $\floor{l}$ (remembering that the packing density is always less than $1$). Unfortunately,
when starting from the M-Lattice or the lattice constructed from a good sandwiching ellipsoid, one has little control over the packing density. In
both cases, $\lambda_1(K,\lat)$ could be as small $1/n$ while the volume of $K$ can be essentially equal to $\det(\lat)$, yielding a packing density
of $n^{-O(n)}$. We note that constructing lattices with packing density $2^{-O(n)}$ is non-trivial even for $\ell_p$ norms (for any fixed $p <
\infty$), where no $\poly(n)$ time computable explicit constructions are known (simple probabilistic constructions do exist, although their
correctness cannot be efficiently verified). As a first simple workaround for the M-Lattice, if one is willing to forgo the packing to covering
property for $K$, then one can simply ``truncate the long parts'' of $K$, replacing $K$ by $K' = K \cap \parl(B)$. Here $\lambda_1(K',\lat) \geq 1$
since $K' \subseteq \parl(B)$, and 
\[
\vol_n(K') \geq \vol_n(K)/N(K,\parl(B)) \geq 2^{-O(n)} \vol_n(K) \geq 2^{-O(n)} \det(\lat).
\]
Therefore the packing density of $\lat$ with respect to $K'$ is $2^{-O(n)}$, and hence Rogers densification procedure creates an easy to enumerate
$3^n$-thin $K'$-covering lattice $\Lambda$ (by the bound of $1/3$ on the packing to covering ratio), which yields a similarly easy to enumerate
$2^{O(n)}$-thin $K$-covering lattice.

We would like to point out that there are probabilistic versions of the densification procedure, which allow us to get around the requirement in
Rogers' greedy construction that we start with a scaling of $K$ that packs with respect to the base lattice. Though we do no use this technique, we describe
it at a high level for the sake of comparison. Roughly speaking, here we densify the base lattice $\lat$ by picking a random superlattice $\Lambda
\subseteq \lat$ of some fixed index. In the case where $\lat = \Z^n$ (which is WLOG by a linear transformation), one prominent such family of
densifications is derived from a random generator matrix of a subspace of $\Z_p^n$. More precisely, letting $C \leftarrow \Z_p^{n \times m}$ be
a uniform $n \times m$ matrix with entries in $\Z_p$, $m \leq n$, we define
\[
\Lambda(C) = \Z^n + (C/p)\Z^m = \Z^n + (C/p)\Z_p^m \text{.} 
\]
Note that with high probability $\Lambda(C)$ corresponds to $p^m$ shifts of $\Z^n$, and hence the index is almost always $p^m$. At a high
level, using random densification, one can show that the index $p^m$ of $\Lambda(C)$, needed for $\Lambda(C)$ to be $K$-covering is proportional to
the \emph{coset volume} of $K$. With respect to $\Z^n$, the \emph{coset volume} $V = \vol(K \imod{\Z^n})$ of $K$, is the volume of the cosets of
$\Z^n$ represented by $K$. This can be identified with the standard Lesbesgue measure by sending each vector $\vecx \in K$ to its fractional parts 
\[
\vecx = (\vecx_1,\dots,\vecx_n) \rightarrow (\vecx_1-\floor{\vecx_1},\dots,\vecx_n-\floor{\vecx_n}) \in [0,1)^n\text{,}
\]
and computing the volume of the resultant set. Here $V/\det(\Z^n) = V$ yields the fraction of cosets represented by $K$. The generator matrix
$C$ can be thought of yielding $p^m$ nearly uniform (though not independent) shifts of $K$ within $[0,1)^n$ (thought of as the torus in the obvious
way).  Here it can be shown that if $p^m = 2^{\Theta(n)}/V$ (for $m$ and $p$ appropriately related), then the shifts $(C/p)\Z_p^m + K$ cover the torus
$[0,1)^n$ after modding by $\Z^n$ with high probability (see \cite{journal/inf-theory/ELZ05} for a similar analysis when $K$ is the Euclidean ball),
and hence $\Lambda(C)$ is $K$-covering with high probability. We note that if $\Z^n$ is $K$-packing then the coset volume of $K$ is $\vol_n(K)$, and
hence we get a bound on $p^m$ in terms of the packing density as in Rogers construction. Also, note that $\vol_n(\parl(B) \cap K)$ is a lower bound on the
coset volume, which allows us to recover the workaround. While these randomly densified lattices are very flexible, apart from the fact they give
probabilistic constructions, they do not seem to give us much advantage for building the thin covering lattices we need here. In particular, for the
parameter range $p^m = 2^{\Theta(n)}$ we require here, it is unclear whether we get better results than with Rogers' greedy construction, and
furthermore the analysis becomes somewhat delicate and non-trivial. We note that $p^m \rightarrow \infty$ (for $p$ and $m$ appropriately related), it
is known that the distribution on the rescaled lattices $p^{m/n} \Lambda(C)$ (so determinant equals $1$) converges in a strong sense to the Haar distribution on
lattices~\cite{goldstein03:_equid_of_hecke_point}. Hence, one can in fact use these distributions to construct lattices that are far more ``extremal''
than what we need or can even hope to certify; in particular, one can recover the results of~\cite{Rogers1959,Butler1972} using these lattices.  

We now explain how to build a thin covering lattice for $K$ with packing to covering ratio at least $1/3$, avoiding the use of the intermediate body
$K'$ above. In the above construction, the truncation $K' = K \cap \parl(B)$ achieves $K' \cap \lat = \set{\veczero}$ and $\vol_n(K)/\vol_n(K') =
2^{O(n)}$. Here the idea will be that, instead of modifying $K$, we will build a sparsifying sublattice $M \subseteq \lat$ which removes all the non-zero
lattice vectors in $K$, i.e.~such that $M \cap K = \set{\veczero}$. As long as the index of $M$ with respect to $\lat$ is at most $2^{O(n)}$,
we will have that $\lambda_1(K,M) = \Theta(1)$. By construction $\lambda_1(K,M) \geq 1$, and Minkowski's convex body theorem 
\[
\lambda_1(K,M) \leq 2 \frac{\det(M)^{1/n}}{\vol_n(K)^{1/n}} = O(1) \frac{\det(\lat)^{1/n}}{\vol_n(K)^{1/n}} = O(1) \text{.}
\]
These bounds will simultaneously guarantee two key properties. Firstly, the iterations in Rogers' greedy construction can be performed by enumerating the
lattice points in $M$ within shifts of $\lambda_1K$, $\lambda_1 = O(1)$, which will have SE enumeration complexity $2^{O(n)}$ ($M$ inherits this from $\lat$).
Second, we will get that the packing density of $M$ with respect to $K$ is $2^{-O(n)}$, and therefore the number of iterations in Rogers' construction
will be bounded by $O(n)$. Hence, we have now reduced the problem of building the thin $K$-covering lattice claimed in
Theorem~\ref{thm:thin-lat-informal}, to the problem of building a sublattice $M \subseteq \lat$ satisfying
\[
[\lat : M] = 2^{O(n)} \quad \text{ and } \quad M \cap K = \set{\veczero} \text{.}
\] 

For the purpose of building $M$, we will make direct use of randomized lattice sparsification techniques, which we subsequently derandomize in
$2^{O(n)}$ time. By applying the transformation $B^{-1}$ to $\lat$ and $K$, we may now assume that $\lat = \Z^n$ and $B = (\vece_1,\dots,\vece_n)$,
where $\parl(B) = [-1/2,1/2]^n$. We will now examine the ``dual'' ensemble associated with densifying superlattice distributions. Here we pick
a uniformly random ``parity check'' matrix $A \leftarrow \Z_p^{m \times n}$, $m \leq n$, where the associated lattice is
\[
\Lambda^\perp(A) = \set{\vecz \in \Z^n: A\vecz \equiv \veczero \imod{p\Z^m}} \text{.}
\] 
We will now examine the above sparsifying distribution when $m = 1$ and $p$ is prime (i.e.~a single random linear equation mod $p$), which correspond
to the so-called Goldstein-Mayer lattices~\cite{goldstein03:_equid_of_hecke_point}.  After normalizing so that their determinant is $1$, as $p
\rightarrow \infty$, Goldstein and Mayer~\cite{goldstein03:_equid_of_hecke_point} show that this distribution converges to the Haar distribution on
lattices (in fact, the convergence result stated for densifying distributions is a consequence of this). We note that the Goldstein-Mayer lattices
have had prior interesting applications in Computer Science: they are a crucial ingredient used to prove hardness of approximation (under randomized reductions) of
the gap version of SVP~\cite{DBLP:journals/jacm/Khot05,DBLP:conf/coco/HavivR06}, and were used to develop a deterministic algorithm for
$(1+\eps)$ approximate CVP under any norm which runs in $2^{O(n)}(1+1/\eps)^n$ time and $2^{n}$ space~\cite{conf/soda/cvp/DK13}. 

We now explain how this sparsifying distribution can be used rather directly to build $M$. Let $S = (K \cap \Z^n) \minuszero$ and let $N = |S|$.
Since $\Z^n$ is an M-lattice for $K$, we know that $N = 2^{O(n)}$, where $N$ can be computed in $2^{O(n)}$ time by SE enumeration of $K \cap \Z^n$
using the standard basis $B$. Let $p$ be any prime such that $N < p < 2N$. Note that $p$ always exists (Bertrand's postulate), and can be computed
deterministically in $2^{O(n)}$ time using trial division (one can also use the standard randomized $\poly(n)$ time Las Vegas algorithm to do this as
well). We now let $M = \Lambda^\perp(\veca)$ where $\veca \leftarrow \Z_p^n$ is chosen uniformly. Clearly $[\Z^n : M] = p = 2^{O(n)}$ (almost surely),
and hence we need only verify that $M \cap K = \set{\veczero} \Leftrightarrow M \cap S = \emptyset$. Take $\vecx \in S$. It is not hard to check that
since $\vecx \neq \veczero$ and $|S| = |(K \cap \Z^n) \minuszero|  < p$, that we must have $\vecx \not\equiv \veczero \imod{p\Z^n}$. Since that $p$ is prime and
$\vecx \not\equiv \veczero \imod{p\Z^n}$, we get that $\pr{\vecx}{\veca} \imod{p}$ is uniformly distributed in $\Z_p$. Therefore 
\[
\Pr_{\veca}[\vecx \in M] = \Pr_{\veca}[\pr{\vecx}{\veca} \equiv \veczero \imod{p}] = 1/p \text{.}
\]
By linearity of expectation, $\E[|M \cap S|] = |S|/p = N/p < 1$. Hence, by the probabilistic method, there exists $M \subseteq \Z^n$ satisfying the
desired requirements. To derandomize the above construction, we apply the method of conditional expectations in a standard way to choose the
coefficients of $\veca$ one at a time (see Lemma~\ref{lem:packing-lat} for full details). 

We remark that the above sparsification and subsequent derandomization is a special case of the deterministic sparsification procedure provided
in~\cite{conf/soda/cvp/DK13}. In their work, the sparsification algorithm is somewhat more complex and less efficient as they additionally guarantee
that the distance of any point $\vecx \in \R^n$ to $M$ is at most an additive $O(1)$ factor larger than its distance to $\Z^n$ under $\|\cdot\|_K$
(technically they only guarantee $|M \cap K| \leq 1000$, however the proof is easily modified to guarantee $M \cap K = \set{\veczero}$ at the cost of
a blowup in the $O(1)$ additive distance error). Interestingly, the sparsification algorithm of~\cite{conf/soda/cvp/DK13} yields another method for
building thin-covering lattices. In particular, given any base lattice $\lat$, one can simply apply the sparsification algorithm to $\mu K$, where
$\mu = \mu(K,\lat)$. Here we recover a sublattice $M \subseteq \lat$, such that $\lambda_1(K,M) \geq \mu$ and $\mu(K,M) \leq \mu(K,\lat) + O(\mu) =
O(\mu)$. Note that $M/(c \mu)$ is a $2^{O(n)}$-thin covering lattice with respect to $K$ (since the packing to covering ratio is $\Omega(1)$), for some
absolute constant $c \geq 1$. Unfortunately, the index $[\lat : M]$ will be roughly $|\mu K \cap \lat|$, which is proportional (up to
$2^{O(n)}$ factors) to the thinness of $\lat/\mu$ as a $K$-covering lattice. Since we only know how to transfer the easy enumeration properties of
$\lat$ to $M$ when $[\lat : M] = 2^{O(n)}$, we can only show that this procedure works if $\lat$ (after rescaling) were already a $2^{O(n)}$-thin
$K$-covering lattice, which is what we were trying to achieve in the first place. While it does not seem directly useful here, we note that this
sparsification procedure implies that any easy to enumerate $2^{O(n)}$-thin $K$-covering lattice can be always transformed into a similarly thin and
easy to enumerate lattice with constant packing to covering ratio (albeit a rather small constant).   

As an aside, both the densifying and sparsifying distributions described above have found quite a few other applications in Computer
Science, mosts notably within Lattice based Cryptography, where they have been used to create cryptographically useful distributions on lattices for
which solving the SVP (and other problems) is average case
hard~\cite{ajtai04:_gener_hard_instan_lattic_probl,DBLP:journals/siamcomp/MicciancioR07,DBLP:journals/jacm/Regev09,DBLP:conf/stoc/Peikert09}.

This completes our description thin covering lattice constructions for symmetric bodies. From the discussion, one can see that our new algorithm
combines the tools from many known constructions, namely, the M-Lattice construction together with lattice sparsification and densification techniques,
in non-trivial ways to create easy to enumerate thin covering lattices. 

In the next section, we will give applications of our thin covering lattice construction, and in the process, we will show how to extend the 
construction to general convex bodies.

\subsection{Applications}
\label{sec:intro-app}

\paragraph{Volume Estimation.} As an application, we give a nearly optimal deterministic algorithms for estimating the volume of
any convex body in the oracle model, which improves on the recent work of~\cite{journal/pnas/DV13}.  
% Here the main important distinction with the classical randomized polynomial time algorithms for volume
% estimation is that the output is guaranteed to be correct.  

\begin{theorem}[Volume Estimation] For a convex body $K \subseteq \R^n$, and any $\eps > 0$, one can compute $V \geq 0$ satisfying
$\vol_n(K) \leq V \leq (1+\eps)^n \vol_n(K)$ in deterministic $2^{O(n)}(1+1/\eps)^n$ time and $\poly(n)$ space. \label{thm:vol-est}
\end{theorem}

Comparing to the known lower bounds, for a deterministic algorithm with access only to membership oracle, B{\'a}r{\'a}ny and F{\"u}redi~\cite{BF87,BF88} showed
that approximating volume to within $(1+\eps)^n$, $0 < \eps < 1$, requires at least $(1+c/\eps)^{n/2}$ oracle queries for some constant $c$. Hence the
algorithm captures the optimal dependence on $\eps$ up to a factor $2$ in the exponent. If we allow randomization, the classical result of Dyer, Frieze
and Kannan~\cite{DFK} gives a polynomial time algorithm for estimating the volume of any convex body to within $(1+\eps)$. In this case, the lower
bounds can be avoided because the volume algorithm is allowed to make $2$-sided error (i.e.~return an estimate that can fall outside the confidence
region on both sides) with small probability. A long standing open problem is to derandomize the volume algorithm in polynomial time when the convex
body $K$ is given explicity (e.g.~if $K$ is a polytope, by its definining inequalities). Even when $K$ is polytope, we note that it is not known how to
recover the above results using other methods (or the methods of~\cite{journal/pnas/DV13}, which are similar).

Comparing to~\cite{journal/pnas/DV13}, we obtain a better dependence on $\eps$, reducing it from $(1+\log^{5/2}(1/\eps)/\eps^2)^n$, and our techniques
work for all convex bodies instead of just symmetric bodies. For the second point, we note if one is only interested in a $4^n$ approximation of
volume, then one can reduce to the symmetric case by replacing $K$ by $K-K$, using the inequality $\vol_n(K-K) \leq \binom{2n}{n} \vol_n(K)$
(see~\cite{RS57}). Hence the main problem is in the asymmetric setting is to obtain efficient $(1+\eps)^n$ approximations. 

From the perspective of techniques, in~\cite{journal/pnas/DV13} they algorithmically implement a technique of Milman~\cite{M86}, known as
\emph{isomorphic symmetrization}, which allows one to compute a body $K'$ whose volume is close to that of $K$ and whose
Banach-Mazur distance to a Euclidean ball can be bounded. From here, they compute the number of integer points inside $K'$ -- after an appropriate
ellipsoidal rounding -- via enumeration to approximate the volume of $K'$. In this context, the closer $K'$ is to an ellipsoid, the sparser one can
make the integer grid while preserving the volume approximation quality. On the other hand, the farther $K'$ is from $K$, the larger ratio
between the volumes of $K'$ and $K$. The approximation algorithm proceeds by a careful tradeoff between these two considerations, essentially giving a
recipe for ``slowing down'' the symmetrization procedure. We note that it is only known how to implement the isomorphic symmetrization
procedure when $K$ is symmetric, which limits the applicability of the above technique to symmetric bodies.

In contrast, for the above algorithm we do not try to modify the body $K$. Instead, we build a ``smarter lattice'' for which lattice point counting is
easy and where the natural lattice point counting estimator yields a good approximation of volume. As one might expect, our approach is based on
building a good covering lattice for $K$. We note that our thin lattice construction depends on the M-ellipsoid construction
of~\cite{journal/pnas/DV13}, and hence one can think of Theorem~\ref{thm:vol-est} as a different way to ``amplify'' the information obtained from $K$ via its
M-ellipsoid. 

We now describe the implementation of the lattice point counting strategy and the adapted covering lattice construction for general convex bodies.
Clearly, a natural first choice for such a lattice would be to use an easy to enumerate thin $K$-covering lattice $\Lambda$. In particular, one would
expect that computing a covering of $K$ by $\eps K$ would allow us to compute a good overestimates of the volume as $\eps \rightarrow 0$, and the
thinness of $\Lambda$ would allow us to do this without enumerating too many points. Now, the lattice points in such a covering would lie in $(K-\eps
K) \cap \eps \Lambda$, and hence the standard averaging argument yields
\[
\E_{\vecx}[|(\eps \Lambda+\vecx) \cap (K-\eps K)|] = \frac{\vol_n(K-\eps K)}{\det(\eps \Lambda)}
\]   
Assuming we could approximate such an average, multiplying out by $\det(\eps \Lambda)$ (which is known), we would get an estimate for $\vol_n(K-\eps
K)$. Since $K$ is asymmetric, the rate of convergence of $\vol_n(K-\eps K)$ to $\vol_n(K)$ is unclear. Regardless, from the same analysis, it is
clearly a better idea to compute a covering of $K$ by $-\eps K$ (if at first slightly counterintuitive), where multiplying by $\det(\Lambda)$, the above
average becomes $\vol_n(K+\eps K) = (1+\eps)^n \vol_n(K)$. Even with this equality however, it is still unclear how one might accurately compute this
average (without making the net extremely fine). A natural question therefore is what unconditional bounds can one get on the estimator
\[
 \det(\eps \Lambda)|(1+\eps) K \cap \eps \Lambda| =  \eps^n \det(\Lambda)|(1+\eps) K \cap \eps \Lambda|~?
\]
Note that so far, we have made no use of fact $\Lambda$ is $K$-covering (and also $-K$-covering by symmetry). Indeed from the covering property, 
one can show that there exists a region $F \subseteq -K$, such that $F$ tiles with respect to $\Lambda$, and hence $\vol_n(F) = \det(\Lambda)$
(see Lemma \ref{lem:points-to-vol} for full details). Using this, one can show the containments 
\[
K  \subseteq  ((1+\eps) K \cap \eps \Lambda) + \eps F \subseteq K + \eps(K-K) \text{ ,}
\]
which yield the estimate
\[
\vol_n(K) \leq \vol_n(((1+\eps) K \cap \eps \Lambda) + \eps F) =  \eps^n \det(\Lambda)|(1+\eps)K \cap \eps \Lambda| \leq \vol_n(K+\eps(K-K)) \text{ ,}
\]
where the middle equality follows from the tiling property of $F$. Therefore, when $\Lambda$ is $K$-covering, $\vol_n(K)$ lower bounds
the natural estimator, though we only get the weak upper bound $\vol_n(K+\eps(K-K))$. 

To get around this issue, we will move away from trying to cover $K$ by scaled copies of $K$ or $-K$. In particular, in the above analysis, we pay a
lot for using an asymmetric covering body. Further complicating the issue, our algorithm for computing covering lattices relies heavily on symmetry of
the covering body. As a workaround, we will try to cover $K$ with a ``large'' symmetric body $K_0$ such that some shift $K_0 \subseteq K-\vect$. If
$\Lambda$ is $K_0$-covering, using the properties of $K_0$ and the same analysis as above we get
\begin{align}
\label{eq:vol-est-eq}
\begin{split}
\vol_n(K) &\leq \eps^n \det(\Lambda) |(K+\eps K_0) \cap \eps \Lambda| \leq \vol_n(K+2\eps K_0) \\
          &\leq \vol_n(K + 2\eps K) = (1+2\eps)^n \vol_n(K)\text{ .}  
\end{split}
\end{align}
Hence by switching the covering strategy, we now can achieve an estimator of acceptable quality. By using the thin covering lattice construction of
Theorem~\ref{thm:thin-lat-informal}, the complexity of this estimator will essentially be controlled by the number of lattice points to be enumerated.
Rearranging Equation \eqref{eq:vol-est-eq}, we get
\[
|(K+\eps K_0) \cap \eps \Lambda| \leq (2+1/\eps)^n \frac{\vol_n(K)}{\det(\Lambda)} = (2+1/\eps)^n \frac{\vol_n(K)}{\vol_n(K_0)} \frac{\vol_n(K_0)}{\det(\Lambda)} \text{ .}
\]
Note that $\vol_n(K_0)/\det(\Lambda) \leq 3^n$, since this is the thinness of the $K_0$-covering. Hence the main ``new'' term is $\vol_n(K)/\vol_n(K_0)$. 
To get the desired complexity bound of $2^{O(n)}(1+1/\eps)^n$, the requirements on $K_0$ are now apparent:
\begin{enumerate}
\item $K_0 \subseteq K-\vect$ for some $\vect \in K$, $K_0$ convex symmetric. 
\item $\vol_n(K) = 2^{O(n)}\vol_n(K_0)$.
\end{enumerate}
Note that from the above analysis, we have reduced volume estimation to the problem of constructing a ``good'' symmetric body $K_0$. The existence
of such a body is well-known in convex geometry, and is directly related to the \emph{Kovner-Besicovitch} (KB) symmetry measure of $K$ (as defined in~\cite{Grunbaum61}):
\begin{equation}
\label{eq:kb-def}
{\rm Sym}_{kb}(K) = \max_{\vecc \in K} \vol_n(K[\vecc])/\vol_n(K)
\end{equation}
where $K[\vecc] = (K-\vecc) \cap (\vecc-K)$. Here it is easy to verify that the bodies $K[\vecc]$ are symmetric, and that any optimal body $K_0$
(i.e.~of maximum volume) must be of the form $K[\vecc]$. For our purposes, we need lower bounds on ${\rm Sym}_{kb}(K)$. In this regard, a
straightforward computation reveals that a uniform point in $K$ yields an average KB value of $2^{-n}$, and hence ${\rm Sym}_{kb}(K) \geq 2^{-n}$.
Furthermore, it was shown in~\cite{MP00} that the centroid of $K$ achieves this lower bound. Therefore, with the aid of random sampling
algorithms over convex bodies, finding a center in $K$ of KB value at least $2^{-n}$ is straightforward. However, our goal here is to obtain
a deterministic algorithm.  

We define a point $\vecc \in K$ to be an $\alpha$-approximate Kovner-Besicovitch point for $K$, $0 < \alpha \leq 1$, if its KB value
$\vol_n(K[\vecc])/\vol_n(K)$ is at least an $\alpha$-factor of ${\rm Sym}_{kb}(K)$. For the purposes of volume estimation, given the above analysis, we note
that even a $2^{-O(n)}$ approximate KB point is sufficient. As our main technical tool in this section, we give an algorithm for
deterministically computing approximate KB points:

\begin{theorem} For any convex body $K \subseteq \R^n$, and any $\eps > 0$, one can compute a $(1+\eps)^{-n}$ approximate Kovner-Besicovitch point
$\vecc \in K$ in deterministic $2^{O(n)}(1+1/\eps)^{2n+1}$ time and $\poly(n)$ space.  
\label{thm:approx-kb} 
\end{theorem}

Using the above theorem, the construction of thin covering lattices for general convex bodies bodies claimed in~\ref{thm:thin-lat-informal} becomes
straightforward. In particular, for the given convex body $K$, we compute a $(6/7)^n$ approximate KB point $\vecc \in K$, and output a thin covering 
lattice for the symmetric body $K[\vecc]$ using the construction from the previous section (see Theorem~\ref{thm:asym-thin-covering} for full details).  

In the context of lattice algorithms over asymmetric norms -- which occur quite readily in the study of Integer Programming
(see~\cite{thesis/D12,conf/soda/cvp/DK13,journal/alg/D13} for example) -- the degree of symmetry of the norm ball $K$,
i.e.~$\vol_n(K[\veczero])/\vol_n(K)$, plays an important role in determining the complexity of solving approximate CVP instances under the associated norm.
It was noticed in~\cite{journal/alg/D13}, that since any convex body $K$ can be transformed into a ``near-symmetric'' norm by centering it at point
$\vecc \in K$ of good KB value, one can in fact solve an approximate version of the Integer Programming problem in single exponential time via a
direct reduction to approximate CVP. This algorithm, in turn, plays an important role in the $2^{O(n)}n^n$ time solver for exact Integer Programming
(IP) from~\cite{thesis/D12}, which gives the fastest known algorithm for IP. For both the above algorithms, points of good KB value were computed by
approximating the centroid of the associated convex bodies, relying on random sampling techniques. As a corollary of Theorem~\ref{thm:approx-kb},
we get a direct derandomization of these results yielding:
\begin{enumerate}
\item A deterministic $2^{O(n)}(1+1/\eps)^n$ and $2^n$ space algorithm for solving $(1+\eps)$-approximate Integer Programming. 
\item A determistic $2^{O(n)} n^n$ time and $2^n$ space algorithm for solving Convex Integer Programs.  
\end{enumerate}
 
We now describe the high level of the algorithm behind Theorem~\ref{thm:approx-kb}. First, by rounding $K$, we may assume that $B_2^n \subseteq K
\subseteq (n+1)n^{1/2} B_2^n$. From here, define the sequence of bodies $K_i = 2^i B_2^n \cap K$ (we note the similarity to the volume algorithm
of~\cite{DFK}), for $i \in \set{0,\dots,T}$, $T = O(\log n)$, where $K_0 = B_2^n$ and $K_T = K$. For each $K_i$, $i \in [T-1]$, we will compute a
$3^{-n}$ approximate KB point $\vecc_i$ for $K_i$ from a $3^{-n}$ approximation KB point $\vecc_{i-1}$ for $K_{i-1}$. Finally, in the last step, from
$K_{T-1}$ to $K_T$, we amplify this to $(1+\eps)^{-n}$ approximation. We note that we may start with $\vecc_0 = \veczero$, since this is the center of
symmetry for $K_0 = B_2^n$.

To compute $\vecc_i$ starting from $\vecc_{i-1}$, we perform the following improvement steps: from our current solution for $\vecc_i$, we build a
covering of $1/2K_i + 1/2\vecc_i$ by $(\eps/2)K_i[\vecc_i]$, and replace $\vecc_i$ with the covering element (which lies in $K_i$) of largest value
(where we compute each the value to within $(1+\eps)^n$). The concavity of $\vol_n(K[\vecc])^{1/n}$ (by Brunn-Minkowski) will allow us to show that
at each step, we improve the objective value by essentially a $(1+c\eps)^n$ factor. Hence $O(1/\eps)$ iterations suffice to construct a near optimal solution.

%In particular, we will construct a thin-covering $K_0$-covering lattice $\Lambda$ we construct a
%$2^{O(n)}$-thin $K$-covering lattice $\Lambda$, and simply count the lattice points in $(1+\eps)K \cap \eps \lat$. The estimates 
%\[
%\vol_n(K) = (1+\Theta(\eps))^n |K \cap \eps \Lambda| \det(\eps \Lambda) \text{,} 
%\]
%
%yield the desired approximation factor and runtime. Unfortunately, the constructions from the previous section quite crucially depend on the symmetry
%of $K$, hence they cannot be applied directly if $K$ is asymmetric. If one only desires a $4^n$ approximation of $\vol_n(K)$ however, then a straightforward
%reduction to the symmetric case is to replace $K$ by $K-K$. It follows from the Rogers-Shephard inequality that $\vol_n(K-K) \leq \binom{2n}{n} \vol_n(K)$. 
%
%For achieving smaller approximation factors, one may wonder if a thin $K-K$-covering lattice yields a thin $K$-covering lattice. Here the main problem
%is understanding how large the ratio $\mu(K,\Lambda)/\mu(K-K,\Lambda)$ can be? Unfortunately, very little is known about this question, from either
%the lower or upper bound perspective (except for a simple upper bound of $n$). We note that getting a bound of better than $n^\eps$ would yield an upper
%bound on the lattice width of lattice free convex bodies of $O(n^{1+\eps} \log n)$ (which controls the complexity of integer programming algorithms
%which branch on hyperplanes), where the current best known bound is $O(n^{4/3}\log^{O(1)}(n))$ due to Rudelson \cite{R00} . Hence even achieving $\eps
%\leq 1/3$ would seem to require new geometric techniques.

\paragraph{{\bf Other Applications.}} As mentioned previously, in~\cite{conf/stoc/ASLV13} the $\eps$-nets as constructed above can be used to give a
PTAS for computing additive $\eps$-Nash equilibria when the sum of the payoff matrices has logarithmic $\eps$-rank, or an additive approximation to the
densest subgraph problem when the adjacency matrix has logarithmic $\eps$-rank. While our $\eps$-net construction uses polynomial instead of
exponential space (improving on their main construction), it does not directly improve the complexity of their algorithms since the $\eps$-nets are only
used for covering problems on $O(\log n)$ dimensions. However, our construction does make their approach more scalable to higher dimensions,
i.e.~where the $\eps$-rank of the matrix of interest is super-logarithmic. 

Constructions for $\eps$-nets also directly lead to algorithms for approximating the norms of general linear operators.  If $T: X \rightarrow Y$
is a linear operator, and $X$ and $Y$ are $n$ and $m$-dimensional normed spaces then a $(1+\eps)$ approximation to $\|T\|_{X \rightarrow Y}$ can be
computed in $O(1+1/\eps)^{\min \set{n,m}}$ time as follows. First, since $\|T^*\|_{Y^* \rightarrow X^*} = \|T\|_{X \rightarrow Y}$, we may assume that
$n \leq m$. From here, we compute an $\eps/2$-net $N_{\eps/2}$ of $B_X = \set{\vecx \in \R^n: \|\vecx\|_X \leq 1}$ (we may identify $X$ with $\R^n$ by
choosing any basis) under $\|\cdot\|_X$, and simply output $\max \set{\|T\vecx\|_Y: \vecx \in N_\eps}$. Another related application is for computing
good polyhedral approximations of a symmetric convex body $K$ (similar results hold for general convex bodies after recentering by a good KB point).
In particular, if we compute a covering $N_{\eps/2}$ of $(1-\eps/2)K^\circ$ with respect to $\eps/2K^\circ$, letting $P = \set{\vecx: |\pr{\vecx}{\veca}|
\leq 1, \veca \in N_{\eps/2}}$, we get a symmetric polytope with at most $2^{O(n)}(1+1/\eps)^n$ facets such that $K \subseteq P \subseteq (1+\eps)K$.  We
note that while the previous statements are all classical mathematical facts, our $\eps$-net construction gives the first efficient algorithmic
implementation for them that works in full generality.

\subsection{Organization} The remainder of the paper is organized as follows. In Section~\ref{sec:prelims}, we present some additional basic concepts
related to convexity and lattices that will be needed in the remainder of the paper. In Section~\ref{sec:thin-lat-constr}, we present the thin lattice
construction for symmetric bodies. Here the main subsections are Section~\ref{sec:schnorr-euchner}, which analyzes the properties of Schorr-Euchner
enumeration, and Section~\ref{sec:constr-steps} which analyzes each individual step of the thin covering lattice construction. Lastly, in
Section~\ref{sec:vol-est}, we give the deterministic volume estimation algorithm as well as the thin covering lattice construction for general convex
bodies. Here the main subsection is Section~\ref{sec:approx-kb}, which describes the algorithm for computing approximate Kovner-Besicovitch points.

\section{Preliminaries}
\label{sec:prelims}

\subsection{Convexity}
\label{sec:prelims-c}

The Brunn-Minkowski inequality states that for measurable sets $A,B \subseteq \R^n$ such that $A+B$ is measurable then 
\[
\vol_n(A+B)^{1/n} \geq \vol_n(A)^{1/n} + \vol_n(B)^{1/n}
\]
We use the notation
\[
V_n = \vol_n(B_2^n) = \frac{\sqrt{\pi}^{~n}}{\Gamma(n/2+1)} = (1+o(1))^n \sqrt{\frac{2\pi e}{n}}^{~n} \text{.}
\]
for the volume of the unit Euclidean ball.

% The Kovner-Besicovitch symmetry measure of $K$ (as defined in~\cite{Grunbaum61}) is
% \[
% {\rm Sym}_{kb}(K) = \max_{\vecc \in K} \vol_n(K[\vecc])/\vol_n(K) \text{,}
% \]
% where $K[\vecc] = (K-\vecc) \cap (\vecc-K)$. We note that ${\rm Sym}_{kb}(K) \in [0,1]$, where ${\rm Sym}_{kb}(K) = 1$ iff $K$ is symmetric.

The following is a powerful bound on the covering numbers due to~\cite{RogersZong97}, which relies on constructions of thin coverings of
space (as described in the previous section).
\begin{theorem} For $A,B \subseteq \R^n$ $n$ dimensional convex bodies
\[
\frac{\vol_n(A-B)}{\vol_n(B-B)} \leq N(A,B) \leq \frac{\vol_n(A-B)}{\vol_n(B)} \Theta^*(B) \text{ ,}
\]
where $\Theta^*(B)$ is the minimal thinness of any covering of space by $B$. In particular, for any $n$-dimensional convex body $B$
\[
\Theta^*(B) \leq n \log n + n \log \log n + 5 n \text{ .}
\] \vspace{-2em}
\label{thm:covering-bounds}
\end{theorem}

An important and deep theorem of Milman~\cite{M86} states that every convex body can be well approximated by an ellipsoid from the perspective of covering.

\begin{theorem}[M-Ellipsoid] There exists a constant $c > 0$, such that for all $n \geq 1$ and any symmetric convex body $K \subseteq \R^n$,  an ellipsoid 
$E \subseteq \R^n$ such that $E,K$ have covering numbers bounded by $(c^n,c^n)$.
\label{thm:m-ell-exist}
\end{theorem}

We note that symmetry is unessential in the above construction, in particular if $K$ is asymmetric, one can replace $K$ by $K-K$ and retrieve a
similar result.

In general, we call an ellipsoid $E$ with single exponential covering numbers with respect to a convex body $K$ an M-ellipsoid of $K$ (though the term
is only somewhat loosely defined). We note that the more standard maximum volume contained ellipsoid (John ellipsoid) and the minimum volume
enclosing ellipsoid (Lowner ellipsoid) of $K$ can be quite far from being M-ellipsoids, in particular their covering numbers can be as high as $n^{\Omega(n)}$.

Recently, it was shown in~\cite{journal/pnas/DV13} that Milman's construction can made fully algorithmic:

\begin{theorem}[M-Ellipsoid Algorithm] Given any symmetric convex body $K$, an ellipsoid $E = E(A) \subseteq \R^n$, such that $E,K$ have
covering numbers bounded by $(c^n, c^n)$, for an absolute constant $c \geq 1$, can be computed in deterministic $\poly(n) 2^n$ time and $\poly(n)$ space.
\label{thm:m-ell-alg}
\end{theorem} 

\paragraph{{\bf Computational Model:}} $K \subseteq \R^n$ is an \emph{$(\veca_0,r,R)$-centered} convex body if $\veca_0+rB_2^n \subseteq K \subseteq
\veca_0+RB_2^n$.  When interacting algorithmically with $K$, we will assume that $K$ is presented by a membership  (or weak membership) oracle $O_K$.
Here a membership oracle $O_K$ on input $\vecx \in \R^n$, outputs $1$ if $\vecx \in K$ and $0$ otherwise. A weak membership oracle takes an extra
parameter $\eps$, where it need only return the correct answer on $\vecx \in \R^n$ if $\vecx \notin \partial K + \eps B_2^n$ (i.e. at distance at
least $\eps$ from the boundary). Most of the algorithms presented in this paper, will require weak membership oracles for bodies derived from $K$
(e.g.~Minkowski sums with other bodies, projections, polar body). However, for the simplicity of the presentation, we will generally ignore the
intracies associated with interacting with weak oracles, as such considerations are by now standard.    

The complexity of our algorithms will be computed in terms of the number of oracle queries and arithmetic operations. In this context, polynomial time
allows for polynomial dependence on dimension and polylogarithmic dependence on the sandwiching parameters, Lipshitz factors, and other related
parameters. We use the notation $\tilde{O}(T(n))$ to suppress $\polylog(T(n))$ terms. 

We state some of fundamental algorithmic tools we will require for convex bodies. The following theorem is yields the classical equivalence between
weak membership and weak optimization~\cite{YN76,GLS} for centered convex bodies. As simple corollaries of this theorem, one can derive weak
membership oracles for all the bodies used in this paper (e.g.~weak membership for Minkowski sums, projections, polars).

\begin{theorem}[Convex Optimization via Ellipsoid Method]
\label{thm:convex-opt}
Let $K \subseteq \R^n$ an $(\veca_0,r,R)$-centered convex body given by a weak membership oracle $O_K$. Let $f:\R^n \rightarrow \R$ 
denote an $L$-Lipshitz convex function given by an oracle that, for every $\vecx \in \Q^n$ and $\delta > 0$, returns a rational number 
$t$ such that $|f(\vecx)-t| \leq \delta$. Then for $\eps > 0$, a rational number $\omega$ and vector $\vecy \in K$ satisfying
\[
\omega - \eps \leq \min_{\vecx \in K} f(\vecx) \leq f(\vecy) \leq \omega 
\]
can be computed in polynomial time.
\end{theorem}

The following algorithm from \cite{GLS}, allows us to deterministically compute an ellipsoid with good ``sandwiching'' guarantees for any centered
convex body $K$. 

\begin{theorem}[Algorithm GLS-Round]
  \label{thm:gls-round}
  Let $K \subseteq R^n$ be an $(\veca_0,r,R)$-centered convex body
  given by a weak membership oracle $\mathrm{O}_K$. Then there is a polynomial time algorithm
  to compute $A \succ 0$, $A \in \Q^{n \times n}$ and $\vect \in \R^n$, such that the ellipsoid $E = E(A)$ satisfies 
  \[
  E + \vect \subseteq K \subseteq n^{1/2}(n+1)E + \vect \text{.}
  \]
\end{theorem}

% \begin{lemma} Let $A \subseteq \R^n$ be a convex body, and let $B \subseteq \R^n$ be a symmetric convex body
% such that $A \subseteq l B$ for some $l > 0$. Then
% \[
% \vol_n(\conv(A \cup B)) \leq 3enl N(A,B) \vol_n(B)
% \]
% \end{lemma}

\subsection{Lattices}
\label{sec:prelim-l}

Let $\lat$ be an $n$-dimensional lattice. A \emph{lattice subspace} $V \subseteq \R^n$ of $\lat$, is linear subspace admitting a basis in $\lat$, i.e.~where
$\dim(V) = \dim(V \cap \lat)$. Note that if $M \subseteq \lat$ is a sublattice of finite index, then the set of lattices subspaces of $M$ and $\lat$ 
are identical. Let $\vecv_1,\dots,\vecv_n$ denote linearly independent vectors. $\vecb_1,\dots,\vecb_n$ is a \emph{directional basis}
of $\lat$ with respect to $\vecv_1,\dots,\vecv_n$ if $\linsp(\vecb_1,\dots,\vecb_i) = \linsp(\vecv_1,\dots,\vecv_i)$ for all $i \in [n]$. Such a
directional basis exists if and only if $\linsp(\vecv_1,\dots,\vecv_i)$ is a lattice subspace of $\lat$ for $i \in [n]$.

For a basis $B$ of $\lat$, define its half open parallelepiped $\parl_\circ(B) = B[-1/2,1/2)^n$. Note that $\parl_\circ(B)$ tiles space with respect
to $\lat$, that is, every point in $\R^n$ is in exactly one lattice shift of $\parl_\circ(B)$. Furthermore, any measurable set $F \subseteq \R^n$
which tiles space with respect to $\lat$ satisfies $\vol_n(F) = \det(\lat)$. For a basis $\vecb_1,\dots,\vecb_n$, we denote its associated
Gram-Schmidt projections by $\pi_1,\dots,\pi_n$, where $\pi_i$ is the orthogonal projection on $\linsp(\vecb_1,\dots,\vecb_{i-1})^\perp$. 

The following is known as Minkowski's convex body theorem:
\begin{theorem}[Minkowski] For an $n$-dimensional lattice $\lat$ and symmetric convex body $K$ 
\[
\lambda_1(K,\lat) \leq 2(\det(\lat)/\vol(K))^{\frac{1}{n}} \text{ .}
\] \vspace{-2em}
\label{thm:minkowski}\end{theorem}

\noindent $K$ is $\alpha$-Schnorr-Euchner enumerable ($\alpha$-SE) with respect to $\lat$ with basis $B=(\vecb_1,\dots,\vecb_n)$ (or just with respect to $B$) if 
\[
\max_{i \in [n], \vect \in \R^n} |\pi_i(K+\vect) \cap \pi_i(\lat)| \leq \alpha\text{,} 
\] 
where $\pi_1,\dots,\pi_n$ are the Gram-Schmidt projections with respect to $B$. 

The following lemma from~\cite{journal/cc/GuMiRe05} states that the covering radius of a lattice can be approximated using a simple explicit point set.

\begin{lemma}\label{lem:cov-approx} Let $K$ and $\lat$ be an $n$-dimensional symmetric convex body and lattice. Then
for any $p \in \N$, 
\[
(1-1/p) \mu(K,\lat) \leq \max_{\vecc \in \lat/p \imod{\lat}} d_K(\lat,\vecc) \leq \mu(K,\lat)
\]
\end{lemma}

\section{Thin Lattice Construction}
\label{sec:thin-lat-constr}

We now describe the three main steps behind the new lattice construction:
\begin{enumerate}
\item {\bf M-Lattice} (Lemma \ref{lem:m-lat}): Construct an M-ellipsoid $E = E(A)$ of $K$ such that $N(K,E) \leq c^n$ and $2^{n+1} \vol(E) \leq \vol_n(K)$. We
pick $\lat$ to have its basis corresponding to the axes of $E$, and scaled so that $\det(\lat) = \vol_n(E)$. 
\item {\bf Packing Lattice} (Lemma \ref{lem:packing-lat}): Compute $N = |K \cap \lat|-1$ via enumeration, and compute a prime $p$ such that $N < p < 2N$. 
Compute a sparsifier $M \subseteq \lat$ such that $[M:\lat] = p$ (essentially, $M$ is a random sublattice of index $p$), satisfying $1 \leq \lambda_1(K,M) \leq c$.
\item {\bf Rogers Lattice} (Lemma \ref{lem:rogers-lat}): Compute $\lambda = \lambda_1(K,\lat)$. Apply Rogers densification procedure to $M$. This computes a super-lattice
$\Lambda$ of $M$, such that $\lambda = \lambda_1(K,\Lambda)$, and where $\mu(K,\Lambda) \leq (3/2) \lambda$. 
Return the $K$-covering lattice $\frac{2}{3 \lambda} \Lambda$. 
\end{enumerate}

The main result of this section is the following lattice construction (which formalizes Theorem \ref{thm:thin-lat-informal} for symmetric bodies):

\begin{theorem} For a symmetric convex body $K \subseteq \R^n$, there is a deterministic $2^{O(n)}$ time and $\poly(n)$ space
algorithm which computes an $n$-dimensional lattice $\Lambda$ with basis $B$ satisfying
\begin{enumerate}
\item $\Lambda$ has a packing to covering ratio of at least $1/3$ with respect to $K$. 
In particular, $\Lambda$ is a $3^n$-thin $K$ covering lattice.
\item $K$ is $2^{O(n)}$-SE with respect to $\Lambda$ with basis $B$. 
\end{enumerate}
Furthermore, for any convex body $C \subseteq \R^n$, the set $(C+K) \cap \Lambda$ can be enumerated in $2^{O(n)} N(C,K)$ time using $\poly(n)$ space.
\label{thm:thin-lat}
\end{theorem}
\begin{proof}
The construction follows by applying Lemmas \ref{lem:m-lat},~\ref{lem:packing-lat},~\ref{lem:rogers-lat} in sequence. The furthermore
follows directly from Lemma \ref{lem:se-enum} since $K$ is $2^{O(n)}$-SE with respect to $\Lambda$ with basis $B$.
\end{proof}

\subsection{Schnorr-Euchner Enumeration}
\label{sec:schnorr-euchner}

We now formalize the implementation of Schnorr-Euchner lattice point enumeration over an $n$-dimensional convex body $K$ and lattice $\lat$ with
basis $B = (\vecb_1,\dots,\vecb_n) \in \R^{n \times n}$. For $i \in \set{0,\dots,n-1}$, define the submatrices 
\[
B^{i-} = (\vecb_1,\dots,\vecb_{n-i}) \quad \text{ and } \quad B^{i+} = (\vecb_{n-i+1},\dots,\vecb_n) \text{,}
\]
and similarly for a vector $\vecx \in \R^n$, we define 
\[
\vecx^{i-} = (\vecx_1,\dots,\vecx_{n-i}) \quad \text{ and } \quad \vecx^{i+} = (\vecx_{n-i+1},\dots,\vecx_n) \text{.}
\]
The enumeration algorithm is presented below (Algorithm~\ref{alg:se}).
 
\begin{algorithm}
\caption{Schnorr-Euchner$(K,B,i,\vecz)$}
\begin{algorithmic}[1]
\ENSURE $(\veca_0,r,R)$-centered convex body $K \subseteq \R^n$ given by a membership oracle, \\
         $\lat(B)$ an $n$-dimensional lattice, level $i$, $0 \leq i \leq n-1$, $B\vecz \in K$, $\vecz^{i+} \in \Z$.   
\REQUIRE Enumeration of $K \cap (\lat(B^{i-}) + B^{i+}\vecz^{i+})$.
\FORALL{$c \in \set{c \in \Z: \exists \vecw \in \R^{n-i-1} \text{ s.t. } B(\vecw,c,\vecz^{i+}) \in K}$} 
\IF{$i = n-1$}
  \STATE Output $B(c,\vecz_2,\dots,\vecz_n)$.
\ELSE
  \STATE Compute $\vecw \in \R^{n-i-1}$ such that $B(\vecw,c,\vecz^{i+}) \in K$.
  \STATE Call Schnorr-Euchner$(K, B, i+1, (\vecw,~ c,~ \vecz^{i+}))$.
\ENDIF
\ENDFOR
\end{algorithmic}
\label{alg:se}
\end{algorithm}

To begin Schnorr-Euchner enumeration on $K$, we call Schnorr-Euchner$(K,B,0,B^{-1}\veca_0)$ (remembering that $K$ is $\veca_0$-centered). The essential
difference with the standard implemention where $K$ is a ball, is the need to solve convex programs in the for loop in line $1$. In particular,
here we must decide for some $c \in \Z$ whether
\begin{equation}
\label{eq:se-conv-prog}
\exists \vecw \in \R^{n-i-1} \text{ s.t. } B(\vecw,c,\vecz^{i+}) \in K \Leftrightarrow \pi_{n-i}\left(B^{(i+1)+}(c, \vecz^{i+})\right) \in \pi_{n-i}(K)
\end{equation}
where $\pi_{n-i}$ is the associated Gram-Schmidt projection of $B$. By the above, we note that the set of $c \in \R$ for which the above condition
holds is a line segment in $\R$ (since it is $1$ dimensional and convex). Hence, the integers $c$ satisfying Equation~\eqref{eq:se-conv-prog} form a
consecutive interval.  Furthermore, by our conditions on the input vector $\vecz \in \R^n$ to the algorithm, the coefficient $z_{n-i}$ lies in this
line segment. Hence, determining all the integer values of $c$ satisfying~\eqref{eq:se-conv-prog} can be enumerated via a line search around $z_{n-i}$ in time
\[
\poly(n)(1 + |\set{c \in \Z: \exists \vecw \in \R^{n-i-1} \text{ s.t. } B(\vecw,c,\vecz^{i+}) \in K}|)
\]
In practice we will only be able to solve the above convex program approximately, i.e. where here we compute a vector $\vecw$ which
approximately minimizes the Euclidean distance between $B(\vecw,c,\vecz^{i+})$ and $K$. We note that this corresponds to building a weak membership
oracle for the line segment. However, even with only a weak oracle, we can easily modify the above algorithm to guarantee that we enumerate the points in
$K \cap \lat$ and perhaps some points in $(K+\eps B_2^n) \cap \lat$. From the perspective of our applications, this is more than sufficient, and
the runtime bounds for the enumeration will be for all intents and purposes identical. We omit the details.  

Lastly, from the above analysis, we get that the choices made at the $i^{th}$ level of recursion, associated with the coefficients of $\vecb_{n-i}$,
are in one to one correspondance with the lattice points
\[
\pi_{n-i}(\lat) \cap \pi_{n-i}(K) \text{.}
\]
From this and the other observations above, we can immediately derive the following lemma (which is standard when $K$ is the Euclidean ball, see for example Lemma
3.1~\cite{DBLP:conf/crypto/HanrotStehle07}), which gives the essential complexity of Schnorr-Euchner enumeration.

\begin{lemma} Let $K \subseteq \R^n$ be a convex body and let $\lat$ be an $n$-dimensional lattice with basis $B$.
Then the lattice points in $K \cap \lat$ can be enumerated (where every point is ouputted exactly once) in time
\[
\poly(n) \sum_{i=1}^n |\pi_i(K) \cap \pi_i(\lat)|
\] 
using $\poly(n)$ space, where $\pi_1$,\dots,$\pi_n$ are the Gram-Schmidt projections of $B$. In particular, if $K$ is $\alpha$-SE
with respect to $\lat$ with basis $B$, then $K \cap \lat$ can be enumerated in $\alpha \poly(n)$ time.
\label{lem:se-enum}
\end{lemma}

In the remainder of the section, we give useful bounds on the Schnorr-Euchner (SE) enumeration complexity.  In particular, we show that SE complexity
can be bounded by the covering number with respect to a fundamental parallelepiped, and that SE complexity behaves well under taking sublattices and
superlattices.

\begin{lemma}
Let $K \subseteq \R^n$ be a convex body and let $\lat$ be a lattice with basis $B$. Then $K$ is $N(K,\parl_\circ(B))$-SE 
with respect to $\lat$ with basis $B$.
\label{lem:cover-to-se}
\end{lemma}
\begin{proof}
Write $B = (\vecb_1,\dots,\vecb_n)$. Let $W_i = \linsp(\vecb_1,\dots,\vecb_{i-1})^\perp$, and $\pi_1,\dots,\pi_n$ be the Gram-Schmidt projections of $B$. 
We must show that for any $\vecx \in \R^n$,
\[
|\pi_i(K+\vecx) \cap \pi_i(\lat)| \leq N(K,\parl_\circ(B)) \text{.}
\]
Let $B_i = (\pi_i(\vecb_i),\dots,\pi_i(\vecb_n))$ for $i \in [n]$. Note that $B_i$ is non-singular, $\pi_i(\lat) = \lat(B_i)$ and that
$\pi_i(\parl_\circ(B)) = \parl_\circ(B_i)$. Let $T \subseteq \R^n$ be an optimal covering of $K$ by $\parl_\circ(B)$, i.e.~$K \subseteq T +
\parl_\circ(B)$ and $|T| = N(K,\parl_\circ(B))$.  Since projections preserve coverings, we also have that 
\[
\pi_i(K+\vecx) \subseteq \pi_i(T + \vecx + \parl_\circ(B)) = \pi_i(T) + \pi_i(\vecx) + \parl_\circ(B_i)
\]
Since $\parl_\circ(B_i)$ tiles $W_i$ with respect to $\pi_i(\lat)$, any shift in of $\parl_\circ(B_i)$ in $W_i$ contains exactly point of
$\pi_i(\lat)$. Hence
\[
|\pi_i(K+\vecx) \cap \pi_i(\lat)| \leq |(\pi_i(T) + \pi_i(\vecx) + \parl_\circ(B_i)) \cap \lat| \leq |\pi_i(T)| \leq |T| = N(K,\parl_\circ(B))
\]
as needed.
\end{proof}

\begin{lemma}
Let $K \subseteq \R^n$ be a convex body which is $\alpha$-SE with respect to an $n$-dimensional lattice $\lat$ with basis $B$. If $M$ is a 
\begin{enumerate}
\item {\bf Full rank sublattice of $\lat$}: $K$ is $\alpha$-SE with respect to $M$ and basis $B_M$, 
\item {\bf Superlattice of $\lat$}: $K$ is $\alpha [M:\lat]$-SE with respect $M$ and basis $B_M$, 
\end{enumerate}
where $B_M$ is a directional basis of $M$ with respect to $B$. Furthermore, if $M$ is given by a basis $H \in \R^{n \times n}$, then
the directional basis $B_M$ can be computed in polynomial time.
\label{lem:se-robust}
\end{lemma}
\begin{proof}
Let $\pi_1,\dots,\pi_n$ denote the Gram-Schmidt projections of $B$. In both cases $1$ and $2$, note that $M$ and $\lat$ have
exactly the same lattice subspaces, and hence a directional basis $B_M$ of $M$ with respect to $B$ exists. Furthermore, by construction both $B$ and
$B_M$ have exactly the same Gram-Schmidt projections.  

For case $1$, the SE complexity bound of $K$ with respect to $M$ with basis $B_M$ is therefore
\[
\max_{\vect \in \R^n} |\pi_i(K+\vect) \cap \pi_i(M)| \leq \max_{\vect \in \R^n} |\pi_i(K+\vect) \cap \pi_i(\lat)| \leq \alpha
\]
by the inclusion $M \subseteq \lat$. For case $2$, we note that we can write $M = S + \lat$, where $|S| = [M:\lat]$ (here $S$ simply chooses one
representative from each coset $M \imod{\lat}$). From here, we see that the SE complexity is bounded by 
\begin{align*}
|\pi_i(K+\vect) \cap \pi_i(M)| 
                          &= |\pi_i(K + \vect) \cap \pi_i(\lat + S)| 
                                       \leq \sum_{\vecs \in S} |\pi_i(K + \vect) \cap \pi_i(\lat + \vecs)| \\ 
                          &= \sum_{\vecs \in S} |\pi_i(K + \vect-\vecs) \cap \pi_i(\lat)| \leq  \alpha|S| =  \alpha[M:\lat] \text{,}
\end{align*}
for any $\vect \in \R^n$, as needed.

We prove the furthermore. Here $M$ is the given by a basis $H$. By solving a system of linear equations we can compute a matrix $X \in \R^{n \times
n}$ such that $H X = B$. 

We claim that $X \in \Q^{n \times n}$. If $M$ is a superlattice of $\lat$, then by inclusion, we clearly have that $X \in \Z^{n \times n}$. If $M$ is
a sublattice, since $\lat \imod{M}$ is an abelian group of order $[\lat:M] = \det(M)/\det(\lat)$, the coefficients of any lattice vector in $\lat$ with respect to $H$ must
be multiples of $1/[\lat:M]$. In particular, the matrix $[\lat:M] X \in \Z^{n \times n}$. This proves the claim.

Now we note that $H$ is a directional basis with respect to $B$ if and only if $X$ is upper triangular. Hence, computing a directional basis is
equivalent to computing an $n \times n$ unimodular matrix $U$ such that $UX$ is upper triangular, since then $HU^{-1}$ is the desired basis.  This can
be achieved by computing the unimodular transformation $U$ which puts $UX$ (or $[\lat:M] UX$) into Hermite Normal Form (HNF). Since the HNF can be computed
in polynomial time, computing a directional basis can be computed in polynomial time as claimed. 
\end{proof}

\begin{lemma} Let $K$ be convex symmetric and $\alpha$-SE with respect to $\lat$ with basis $B$. Then for any convex body $C \subseteq \R^n$, $C$ is
$\alpha N(C,K)$-SE with respect to $\lat$ with basis $B$. Furthermore, for the body $C+K$ this bound specializes to $O(\alpha 3^n (n \log n)N(C,K))$.
\label{lem:se-cover-bounds} 
\end{lemma}
\begin{proof}
Let $T$ be an optimal covering of $C$ by $K$. Then
\[
|\pi_i(C) \cap \pi_i(\lat)| \leq |\pi_i(T+K) \cap \pi_i(\lat)| \leq \sum_{\vect \in T} |\pi_i(\vect + K) \cap \pi_i(\lat)| \leq \alpha |T| = \alpha N(C,K) \text{,}
\] 
as needed. For the furthermore, it follows from the inequality
\begin{align*}
N(C+K,K) &\leq N(C,K) N(2K,K) = O((n \log n) \vol_n(2K+K)/\vol_n(K)) N(C,K) \\ 
         &= O(3^n(n \log n) N(C,K))
\end{align*}
where the last inequality follows from Theorem~\ref{thm:covering-bounds}.
\end{proof}

\subsection{Construction Steps}
\label{sec:constr-steps}

\begin{lemma}[M-Lattice] Let $K$ be a symmetric convex body. There is a deterministic $2^{O(n)}$ time and $\poly(n)$ space algorithm which computes a
lattice $\lat$ with basis $B$, satisfying
\[
1.~2^{n+1} \det(\lat) \leq \vol_n(K) \quad \quad 2.~N(K,\parl_\circ(B)) \leq c^n
\]
for some absolute constant $c \geq 1$. In particular, $K$ is $c^n$-SE with respect to $\lat$ with basis $B$.    
\label{lem:m-lat}
\end{lemma}
\begin{proof}
Using Theorem~\ref{thm:m-ell-alg} we compute an $M$-ellipsoid $E = E(A)$ for $K$, such that $K,E$ have covering numbers bounded by $(c_0^n,c_0^n)$.
This can be done deterministically in $2^{O(n)}$ time and $\poly(n)$ space.

Let $B = 1/(2^{1+1/n}c_0) V_n^{1/n} A^{-1/2}$. We claim that $\lat = \lat(B)$ satisfies the desired properties. First, we remember that $E = A^{-1/2}
B_2^n$ and that $\vol_n(E) = |\det(A^{-1/2})| V_n$. 

For property $1$, we have that
\[
2^{n+1} \det(\lat) = \det(B) = 2^{n+1}(2^{-(n+1)}c_0^{-n} V_n |\det(A^{-1/2})|) = c_0^{-n} \vol_n(E) \leq \vol_n(K)\text{,}
\]
as needed, where the last inequality follows from the fact that $\vol_n(E) \leq N(E,K) \vol_n(K)$.

For property $2$, we first note that 
\[
\parl_\circ(B) = A^{-1/2} \Bigg[-\frac{V_n^{1/n}}{2^{2+1/n}c_0},\frac{V_n^{1/n}}{2^{2+1/n}c_0}\Bigg)^n \text{.} 
\]
Assuming that $c_0 \geq 2$ (it is actually much larger), it is easy to see that $V_n^{1/n}/(2^{2+1/n}c_0) \leq 1/\sqrt{n}$ (at least for
$n$ large enough) since $\sqrt{n} V_n^{1/n} \rightarrow \sqrt{2\pi e} \leq 5$. Therefore we may assume that 
\[
A^{-1/2} \Bigg[-\frac{V_n^{1/n}}{2^{2+1/n}c_0},\frac{V_n^{1/n}}{2^{2+1/n}c_0}\Bigg)^n \subseteq 
          A^{-1/2} \Bigg[\frac{-1}{\sqrt{n}},\frac{1}{\sqrt{n}}\Bigg)^n \subseteq A^{-1/2}  B_2^n = E(A) \text{.}
\]
From here, we have that
\[
N(K,\parl_\circ(B)) \leq N(K,E) N(E,\parl_\circ(B)) \leq c^n N(E,\parl_\circ(B))
\]
Using the fact that $\parl_\circ(B)$ tiles space with respect to $\lat$ (and hence has covering density $1$), we get that 
\begin{align*}
N(E,\parl_\circ(B)) &\leq \frac{\vol_n(E-\parl_\circ(B))}{\vol_n(\parl_\circ(B))} \leq \frac{\vol_n(2E)}{\det(\lat)} = 2^n (2^{n+1} c_0^n) = 2(4c_0)^n
\end{align*}
Putting everything together, we get $N(K,\parl_\circ(B)) \leq 2(4c_0^2)^n \leq c^n$ (for $c = 5c_0^2$ say), as needed. Since the computation of $B$
can be done in $\poly(n)$ time, the desired bound on the runtime and space usage holds. Lastly, for the furthermore, we note that it follows directly
from Lemma \ref{lem:cover-to-se}
\end{proof}

\begin{lemma}[Packing Lattice] Starting from $\lat$ and $B$ be as in Lemma \ref{lem:m-lat},
a sublattice $M \subseteq \lat$, $[\lat : M] \leq 2c^n$, and its directional basis $B_M$ with respect to $B$, satisfying $1 \leq \lambda_1(K,M) \leq c$ 
can be computed in deterministic $\poly(n) c^{2n}$ time using $\poly(n)$ space. Furthermore, $K$ is $c^n$-SE with respect to $M$ with basis
$B_M$, and $M$ has packing density at least $c^{-n}$ with respect to $K$.  
\label{lem:packing-lat}
\end{lemma}
\begin{proof}
By a change of basis, that is multiplying by $B^{-1}$, we may assume that $\lat = \Z^n$ and that our basis is $\vece_1,\dots,\vece_n$.  We shall first
show the existence of $M$ via the probabilistic method ($M$ will be a random sublattice of $\lat$), and then use the method of conditional
expectations to derandomize the construction.  

\paragraph{{\bf Existence:}} Let $S  = (K \cap \Z^n) \minuszero$, and let $N = |S|$. Since $K$ is symmetric 
\[
\vol_n(K) \geq 2^{n+1} \det(\Z^n) > 2^n \text{,}
\]
by Minkowski's convex body theorem we know that $N \geq 2$. Let $p$ be a prime such that $N < p < 2N$ (that such a prime always exists is Bertrand's
postulate). 

\begin{claim}  $\forall~ \vecx \in S$, $\vecx \not\equiv \veczero \imod{p\Z^n}$. \label{cl:p1} \end{claim} 
\begin{proof}
For the sake of contradiction, assume that for some $\vecx \in S$, $\vecx \equiv \veczero \imod{p\Z^n}$. Then by convexity and symmetry of $K$,
we must have that $\pm \set{\vecx/p,2\vecx/p,\dots,\vecx} \subseteq K \cap \Z^n \minuszero = S$. But then $|S| \geq 2p$, a clear contradiction.
\end{proof}

Let $\veca \leftarrow \Z_p^n$ be a uniform element of $\Z_p^n$. Let $M = \set{\vecy \in \Z^n: \pr{\veca}{\vecy} \equiv 0 \imod{p}}$. Note
that as long as $\veca \neq \veczero$ (in this case $M = \Z^n$), $M$ is a sublattice of $\Z^n$ of index $[\Z^n : M] = p$.

\begin{claim} $\E_{\veca}[|(M \cap K) \minuszero|] = N/p < 1$. \label{cl:p2} \end{claim}
\begin{proof}
Since for all $\vecx \in S$, $\vecx \not\equiv \veczero \imod{p\Z^n}$ (by Claim \ref{cl:p1}), we have that $\pr{\vecx}{\veca}$ is uniformly distributed in
$\Z_p$ since $p$ is prime. In particular, $\Pr_\veca[\pr{\veca}{\vecx}] = 1/p$. Therefore by linearity of expectation
\[
\E_{\veca}[|(M \cap K) \minuszero|] = \sum_{\vecx \in S} \Pr_\veca[\vecx \in M] = \sum_{\vecx \in S} \Pr_\veca[\pr{\vecx}{\veca} \equiv 0 \imod{p}]  
                                    = \sum_{\vecx \in S} 1/p = N/p < 1
\]
\end{proof}

By Claim \ref{cl:p2}, there exists $\veca \in \Z_p^n$ such the associated lattice $M$ satisfies $|(M \cap K) \minuszero| = 0$. We show that
$M$ satisfies the conditions of the lemma. First, by construction, we have $\lambda_1(K,M) \geq 1$. The following claim yields the upper bound:

\begin{claim} For $M$ as above, we have that $\lambda_1(K,M) \leq c$. \label{cl:p3} \end{claim}
\begin{proof}
Firstly, note that
\[
\det(M) \leq p < 2N \leq 2 |K \cap \lat| \leq 2 c^n
\]
Next, by construction 
\[
\vol(c K) = c^n \vol(K) \geq c^n 2^{n+1} \geq 2^n \det(M) \text{.}
\]
Hence by Minkowski's convex body theorem, $\lambda_1(K,M) \leq c$ as needed. 
\end{proof}

Let $\lambda = \lambda_1(K,M)$. We can lower bound the packing density of $M$ with respect to $K$ as follows:
\[
\frac{\vol_n(\lambda/2K)}{\det(M)} \geq \frac{\vol_n(1/2K)}{p} = \frac{2^{-n} \vol_n(K)}{p} \geq \frac{2}{p} \geq \frac{2}{2c^n} = c^{-n} \text{,}
\]
as needed. Lastly, that $K$ is $c^n$-SE with respect to $M$ with basis $B_M$ follows directly from Lemma \ref{lem:se-robust} and the guarantee
that $K$ is $c^n$-SE with respect to $\Z^n$ with the standard basis.

\paragraph{{\bf Algorithm:}} We now show how to derandomize the above construction in $\poly(n) c^{2n}$ time using only $\poly(n)$ space.  
The idea here is simply to choose the coefficients of $\veca = (a_1,\dots,a_n)$ one at a time from left to right. Each time we fix a coefficient
we will guarantee that conditioned on fixed coefficients, the expected number of points in $M \cap K \minuszero$ (averaging over the
randomness for the remaining coefficients) is less than $1$. We now give the formula for the conditional expectation. For a vector $\vecx \in \R^n$,
define 
\[
\vecx^{i-} = (\vecx_1,\dots,\vecx_i) \quad \text{ and }  \vecx^{i+} = (\vecx_{i+1},\dots,\vecx_n) \text{.}
\]
Assume we have already fixed $\veca^{(i-1)-} = (c_1,\dots,c_{i-1})$ and are left with choosing the values of $a_i,\dots,a_n$. 
If we set $a_i = c_i$, we condition $\veca$ on the event $\veca^{i-} = (c_1,\dots,c_i) \eqdef \vecc^i$. Then we have that
\begin{align}
\label{eq:p1}
\begin{split}
\E_\veca[|(M \cap K) \minuszero| ~|~ \veca^{i-} = \vecc^i] 
             &= \sum_{\vecx \in S} \Pr_\veca\left[\pr{\veca}{\vecx} \equiv 0 \imod{p} ~|~ \veca^{i-} = \vecc^i \right] \\
             &= \sum_{\vecx \in S} \Pr_\veca\left[\pr{\vecc^i}{\vecx^{i-}} + \pr{\veca^{i+}}{\vecx^{i+}} \equiv 0 \imod{p}\right]
\end{split}
\end{align}
From here, we have that 
\[
\Pr_\veca\left[\pr{\vecc^i}{\vecx^{i-}} + \pr{\veca^{i+}}{\vecx^{i+}} \equiv 0 \imod{p}\right] = 
                                   \begin{cases} 1/p &:~   \vecx^{i+} \not\equiv \veczero \imod{p\Z^{n-i}} \\ 
                                                   0 &:~   \pr{\vecc^i}{\vecx^{i-}} \not\equiv 0 \imod{p} \\
                                                   1 &:~   \text{otherwise} \end{cases}
\]
Therefore the expectation in Equation~\eqref{eq:p1} can be expressed as
\begin{align}
\label{eq:p2}
\begin{split}
\E_\veca[|(M \cap K) \minuszero| ~|~ \veca^{i-} = \vecc^i]  
                     = ~&|\set{\vecx \in S: \vecx^{i+} \equiv \veczero \imod{p\Z^{n-i}}, \pr{\vecc^i}{\vecx^{i-}} \equiv 0 \imod{p}}|~+ \\
                       ~&|\set{\vecx \in S: \vecx^{i+} \not\equiv \veczero \imod{p\Z^{n-i}}}|/p
\end{split}
\end{align}
Notice that this expectation is less than $1$ if and only if the first set on the right hand side is empty (this set corresponds to the elements that
are definitively in $M$). Since the global expectation is $N/p < 1$, by the properties of conditional expectations and Equation \ref{eq:p2}, we can
guess the coordinates of $\veca$ one by one as long as the set of points definitively in $M$ remains empty (i.e.~the greedy strategy works). 
\vspace{1em}

\noindent From these observations, we get the following algorithm for building $M$:
\begin{algorithmic}[1]
\STATE Compute $N = |S|$ via Schnorr-Euchner enumeration over $K \cap \lat$ (using the standard basis). \\ 
       Pick a prime $p$ satisfying $N < p < 2N$.
\FORALL{$i \in 1 ~\TO~ n$}  
\STATE Guess $\veca_i$ by trying all numbers in $\set{0,\dots,p-1}$. Accept a guess for $\veca_i$ if 
 \[  \set{\vecx \in S: \vecx^{i+} \equiv \veczero \imod{p\Z^{n-i}}, \pr{\veca^{i-}}{\vecx^{i-}} \equiv 0 \imod{p}} = \emptyset \text{.} \]
 Verify this condition for each potential guess using Schnorr-Euchner enumeration over $S$.
\ENDFOR
\RETURN $M = \set{\vecx \in \Z^n: \pr{\vecx}{\veca} \equiv 0 \imod{p}}$
\end{algorithmic}

Given $M$ from the above algorithm, we must still compute a directional basis with respect to the standard basis. This is straightforward. 
Let $j \in [n]$ denote first non-zero coefficient of $\veca$. Rescaling by $\veca_j^{-1} \imod{p}$, we may assume that $\veca_j = 1$. From here,
it is direct to verify that 
\[
(\vece_1,\dots,\vece_{j-1}, p \vece_j, - \veca_{j+1} \vece_j + \vece_{j+1}, \dots, -\veca_n \vece_j + \vece_n)
\] 
is a valid directional basis for $M$.

Since the correctness of the above algorithm has already been argued, it remains to bound the algorithms complexity.  Firstly, by construction $\lat$,
we have that $K$ is $c^n$-SE with respect to $\Z^n$ with the standard basis. Hence, by Lemma \ref{lem:se-enum} every Schnorr-Euchner enumeration over
$K \cap \Z^n$ can be performed in $\poly(n) c^n$ time using $\poly(n)$ space. We perform one such enumeration to compute $N$, and at most $n p \leq 2
n c^n$ such enumerations during the main loop of the algorithm.  Hence the amount of time spent during the enumeration steps is at most $\poly(n)
c^{2n}$. Lastly, the time to compute $p$ is can be bounded by $\poly(n) c^n$, by simply enumerating over all the choices between $N$ and $2N$ and
using any deterministic primality test. 
\end{proof} 

\begin{lemma}[Rogers Lattice] Starting from $M$ and $B_M$ be as in Lemma \ref{lem:packing-lat}, a super-lattice $\Lambda$ of $M$,
with directional basis $B_\Lambda$ with respect to $B_M$, satisfying 
\begin{enumerate}
\item $\lambda_1(K,M) = \lambda_1(K,\Lambda)$,
\item $\mu(K,\Lambda) \leq 3/2 \lambda_1(K,\Lambda) \leq 3c/2$,
\item $[\Lambda : M] \leq c^n$,
\end{enumerate}
can be computed in $\tilde{O}((2c^3)^n)$ time and $\poly(n)$ space. Letting $\lambda = \lambda_1(K,\Lambda)$, we furthermore have that
\begin{enumerate}[(a)]
\item $2/(3\lambda) \Lambda$ is a $3^n$-thin $K$-covering lattice.
\item $K$ is $\tilde{O}((2c^3)^n)$-SE with respect to $2/(3\lambda) \Lambda$ with basis $2/(3\lambda) B_\Lambda$.
\end{enumerate}
\label{lem:rogers-lat}
\end{lemma}
\begin{proof}
To build the covering lattice for $K$ claimed by the Lemma we will use Rogers densification procedure. We first describe and analyze its the basic
properties, then analyze its effects on $M$, and lastly discuss the details of making it algorithmic in our setting. This densification can be applied to any
$n$-dimensional lattice $\lat$. It proceeds as follows:
\vspace{1em}

\noindent Find a coset $\lat + \vecc \in \lat/3 \imod{\lat}$, such that $d_K(\lat, \vecc) > \lambda_1(K,\lat)$. If none exists, return $\lat$. Otherwise,
replace $\lat$ by $\lat + \set{-\vecc,\veczero,\vecc}$, where $\vecc$ is the coset found by the procedure, and repeat. 
\vspace{1em}

\paragraph{{\bf Basic Properties:}} We analyze the properties of $\lat$ at termination. Let $\lambda = \lambda_1(K,\lat)$. By construction, after termination, we must have that
\[
\max_{\vecc \in \lat/3 \imod{\lat}} d_K(\lat,\vecc) \leq \lambda \text{.}
\]
Therefore, by Lemma \ref{lem:cov-approx}, we must have that $\mu(K,\lat) \leq 3/2 \lambda$. We claim that $2/(3\lambda)\lat$ is a $3^n$-thin
$K$-covering lattice. Clearly, $\mu(K, 2/(3\lambda)\lat) \leq 1$ by the previous inequality. For the thinness, note that
\[
\frac{\vol_n(K)}{\det(2/(3\lambda)\lat)} = \frac{\vol_n(3\lambda/2K)}{\det(\lat)} \leq \frac{\vol_n(3\lambda/2K)}{\vol_n(\lambda/2K)} = 3^n
\]
where the inequality $\vol_n(\lambda/2K) \leq \det(\lat)$ follows directly from Minkowski's convex body theorem.

We now bound the convergence time of the densification procedure. We claim that at each non-terminating iteration, the length of the shortest-nonzero vector is unchanged,
while the determinant of $\lat$ decreases by a factor $3$. For the first property, take $\lat + \vecc \in \lat/3 \imod{\lat}$ such that
$d_K(\lat, \vecc) \geq \lambda_1(K,\lat)$. Since $3\vecc \in \lat$, note that $\lat' \eqdef \lat + \Z \vecc = \lat + \set{-\vecc,\veczero,\vecc}$. From here, we
have that
\[
\lambda_1(\lat') = \min \set{d_K(\lat,-\vecx), \lambda_1(K,\lat), d_K(\lat,\vecx)} = \min \set{\lambda_1(K,\lat), d_K(\lat,\vecx)} = \lambda_1(K,\lat) \text{,}
\]
where the equality $d_K(\lat,-\vecx) = d_K(\lat,\vecx)$ follows by symmetry of $K$. Hence the length of the shortest non-zero vector stays unchanged.
The second claimed property follows from $|\lat' \imod{\lat}| = |\Z_3| = 3$.

Let $\alpha = \vol_n(\lambda/2 K)/\det(\lat)$ denote the packing density of $\lat$. By the previous analysis, at each non-terminating iteration, the
packing density of $\lat$ increases by a factor $3$. Since the packing density never exceeds $1$, if $k$ is the number of non-terminating iterations,
we must have that $\alpha 3^k \leq 1 \Rightarrow k \leq \floor{\log_3(1/\alpha)}$.  In particular, if the base lattice is $\lat$ and $\lat_k$ is the
final outputted lattice, we must have that $[\lat_k : \lat] \leq 1/\alpha$.

\paragraph{{\bf Behavior on $M$}:} Let $M$ be the lattice from \ref{lem:packing-lat} with basis $B_M$, and let $\Lambda$ be the lattice outputted by
the densification procedure. Let $\lambda = \lambda_1(K,M)$. Since we are guaranteed that $\lambda_1(K,\Lambda) = \lambda \leq c$, we have that 
$\mu(K,\Lambda) \leq 3/2\lambda \leq 3/2c$. The remaining thinness and covering properties of $\Lambda$ are now guaranteed by the our previous
analysis. Furthermore, since $M$ has packing density at least $c^{-n}$, our previous analysis also ensures that $[\Lambda : M] \leq c^n$. 

Let $B_\Lambda$ denote the directional basis of $\Lambda$ with respect to $B_M$. Since $K$ is $c^n$-SE with respect to $M$ with basis $B_M$, we get
from Lemma \ref{lem:se-robust} that $K$ is $c^n [\Lambda:M] \leq c^{2n}$ SE with respect to $\Lambda$ with basis $B_M$. From Lemma
\ref{lem:se-cover-bounds}, we get that $3c/2K$ is $c^{2n} N(3c/2K,K)$-SE with respect to $\Lambda$ with basis $B_\Lambda$. By Theorem
\ref{thm:covering-bounds}, we get that
\[
N(3c/2K,K) = O(n \log n) \frac{\vol_n(3c/2K + K)}{\vol_n(K)} = O(n \log n (3c/2+1)^n) = \tilde{O}((3c/2+1)^n) \text{.}
\]
Hence $c^{2n} N(3c/2K,K) = \tilde{O}((3c^3/2+c^2)^n) = \tilde{O}((2c^3)^n)$. Since $3/2 \lambda \leq 3/2c$ the same SE holds for $3/2 \lambda K$, and
by scaling for $K$ with respect to $2/(3\lambda) \Lambda$ with basis $2/(3\lambda) B_\Lambda$. Hence $\Lambda$ satisfies all the requirements of the
lemma. 

\paragraph{{\bf Algorithm}:} We analyze the complexity of making Roger's densification algorithmic on $M$. Firstly, we need to compute
$\lambda = \lambda_1(K,M)$. Since $\lambda \leq c$, it suffices to enumerate the points in $cK \cap M$, and return the length of shortest non-zero
vector found. Since $K$ is $c^n$-SE with respect to $M$ with basis $B_M$, by Lemma \ref{lem:se-enum} this enumeration takes at most
\[
\poly(n) c^n N(cK,K) \leq \poly(n) c^n (c+1)^n  \leq \poly(n)(2c^2)^n
\]
time and $\poly(n)$ space. Now let $M_k$ with directional basis $B_{M_k}$ with respect to $B_M$ denote the resultant lattice after $k$ iterations.
Here, for each coset 
\[
M_k + \vecc \in M_k/3 \imod{M_k} = \set{M_k + B_{M_k}\veca/3: \veca \in \set{-1,0,1}^n} \text{,}
\]
we must verify whether $d_K(M_k,\vecc) > \lambda$. Note that this last step is equivalent to checking whether 
\[
d_K(M_k,\vecc) > \lambda \quad \Leftrightarrow \quad M_k \cap (\vecc + \lambda K) = \emptyset \text{,}
\]
which can be verified by straightforward enumeration. Since $[M_k : M] \leq c^n$, by Lemmas \ref{lem:se-robust} and \ref{lem:se-enum} we get the
Schnorr-Euchner enumeration over $\vecc + \lambda K$ takes at most $\poly(n) c^n (2c^2)^n = \tilde{O}((2c^3)^n)$ time and $\poly(n)$ space.  Since we
may enumerate over all $3^n$ cosets of $M_k/3 \imod{M_k}$, the time for a single iteration can be bounded by $\tilde{O}((6c^3)^n)$ time.  Furthermore,
if coset $\vecc$ is to be added to $M_k$, a directional basis for $M_{k+1} = M_k + \Z \vecc$ can clearly be computed in polynomial time from $\vecc$
and $B_{M_k}$. Lastly, since the number of iterations is bounded by $\log_3 c^n = O(n)$, the total runtime can be bounded by $\tilde{O}((6c^3)^n)$ and
the space usage by $\poly(n)$ as needed.
\end{proof}

\section{Volume Estimation}
\label{sec:vol-est}

In this section, we describe the new algorithm for volume estimation. Our algorithm will rely on a construction for thin covering lattices for
general convex bodies bodies, which will in turn rely on an algorithm for computing approximate Kovner-Besicovitch points. The guarantees for the generalized
thin lattice construction (which formalizes Theorem~\ref{thm:thin-lat-informal} for general convex bodies) are as follows:

\begin{theorem}[General Thin Lattice] For a convex body $K \subseteq \R^n$, there is a $2^{O(n)}$ time and $\poly(n)$ space algorithm which computes an
$n$ dimensional lattice $\Lambda$ with basis $B$, and a point $\vecc \in K$ satisfying 
\begin{enumerate}
\item $\Lambda$ is a $3^n$-thin $K[\vecc]$-covering and a $7^n$-thin $K$-covering lattice.
\item $\Lambda$ has packing to covering ratio at least $1/3$ with respect to $K[\vecc]$.
\item $K[\vecc]$ and $K$ are both $2^{O(n)}$-SE with respect to $\Lambda$ with basis $B$.
\end{enumerate}
Furthermore, for any convex body $C \subseteq \R^n$, the set $(C-K) \cap \Lambda$ can be enumerated in $2^{O(n)} N(C,K)$ time using $\poly(n)$ space.
\label{thm:asym-thin-covering}
\end{theorem}
\begin{proof}
We first use algorithm of Theorem \ref{thm:approx-kb} to compute $(6/7)^n$ approximate Kovner-Besicovitch point $\vecc \in K$.  Using the algorithm of
Theorem~\ref{thm:thin-lat} we build a $3^n$-thin $K[\vecc]$-covering lattice $\Lambda$ with basis $B$. Since $K[\vecc] \subseteq K-\vecc$, 
$\Lambda$ is also a $K$-covering lattice. To bound the thinness with respect to $K$, by the guarantees on $\vecc$, we have that
\begin{align*}
\frac{\vol_n(K)}{\det(\Lambda)} &= \frac{\vol_n(K)}{\vol_n(K[\vecc])} \frac{\vol_n(K)}{\det(\Lambda)} \leq \frac{3^n}{(6/7)^n{\rm Sym}_{kb}(K)} \\
                                &\leq (7/6)^n 2^n 3^n = 7^n
\end{align*}
From the guarantees on $\Lambda$, we know that $K[\vecc]$ is $2^{O(n)}$-SE with respect to $B$. Therefore, by Lemma \ref{lem:se-cover-bounds},
the SE complexity of $K$ with respect to $B$ is bounded by
\begin{align*}
2^{O(n)} N(K,K[\vecc]) &= 2^{O(n)} ~O(n \log n)~ \frac{\vol_n(K + K[\vecc])}{\vol_n(K[\vecc])} \leq 2^{O(n)} \frac{\vol_n(2K)}{\vol_n(K[\vecc])} \\
                       &= 2^{O(n)} ~2^n~ \frac{\vol_n(K)}{\vol_n(K[\vecc])} \leq 2^{O(n)} ~2^n~ (7/3)^n = 2^{O(n)}
\end{align*}
as needed. The remaining guarantees on $\Lambda$ and the complexity bound for the above algorithm now follows directly from
guarantees in Theorems~\ref{thm:approx-kb} and \ref{thm:thin-lat}.
\end{proof}

We will use Theorem~\ref{thm:asym-thin-covering} within the volume estimation algorithm. The following Lemma is used to justify the accuracy of
volume estimation algorithm. 

\begin{lemma} Let $K_0, K$ be $n$ dimensional convex bodies. 
Let $\lat$ be an $n$-dimensional $K_0$-covering lattice. For $\eps > 0$, the following holds:
\[
\vol_n(K) \leq \eps^n \det(\lat) ~|\eps\lat \cap (K-\eps K_0)| \leq \vol_n(K+\eps(K_0-K_0)) \text{ .}
\]
Furthermore, if $K_0 \subseteq K-\vecc$, for some $\vecc \in \R^n$, and $K_0$ is symmetric then
\[
\vol_n(K) \leq \eps^n \det(\lat) ~|\eps\lat \cap ((1+\eps)K-\eps \vecc)| \leq (1+2\eps)^n \vol_n(K) \text{ .}
\] \label{lem:points-to-vol} \vspace{-2em}
\end{lemma}
\begin{proof}
\begin{claim} There exists a subset $F \subseteq K_0$ such that $F$ tiles with respect $\lat$. In particular, $\vol_n(F) = \det(\lat)$. \end{claim}
\begin{proof}
Since the tiling~/~covering property is shift invariant, we may shift $K_0$ so that $\veczero$ is in the interior of $K$. From here, note
that $\|\cdot\|_{K_0}$ is an asymmetric norm. We define $F$ to be all the points $\vecx \in \R^n$ such that $\veczero$ 
is the lexicographically minimal closest lattice vector to $\vecx$ under $\|\cdot\|_{K_0}$. More presicely, $\vecx \in F$ iff 
\[
\|\vecx-\veczero\|_{K_0} = \|\vecx\|_{K_0} = d_{K_0}(\lat,\vecx) = \min_{\vecy \in \lat} \|\vecx-\vecy\|_{K_0}
\]
and $\veczero$ is the lexicographically smallest minimizer for the last expression on the right hand side. Since every point in $\R^n$
has a unique lexicographically closest lattice vector in $\lat$, and since the standard lexicographic order on $\R^n$ is shift invariant,
we see that $F$ tiles space with respect to $\lat$. That $\vol_n(F) = \det(\lat)$ follows directly from the tiling property. 

We claim that $F \subseteq K_0$. Assume not, then $\exists \vecx \in F$ such that $\|\vecx\|_{K_0} > 1$. Since
$\lat$ is $K_0$-covering, there exists $\vecy \in \lat$ such that $\vecx \in \vecy + K_0$. But then $\|\vecx-\vecy\|_{K_0} \leq 1 < \|\vecx\|_{K_0}$,
which contradicts that $\veczero$ is a closest lattice vector to $\vecx$. Hence $F \subseteq K_0$ as claimed.
\end{proof}

Since $\eps\lat$ is $\eps F$-tiling (where $F$ is as above), we have that the $\eps \lat$ shifts of $\eps F$ covering $K$ correspond exactly
to the centers $\eps \lat \cap (K - \eps F)$. From here, since $F \subseteq K_0$, we have the inclusions 
\begin{align}
\label{eq:vest-eq1}
\begin{split}
K &\subseteq (\eps \lat \cap (K - \eps F)) + \eps F \\
  &\subseteq (\eps \lat \cap (K - \eps F)) + \eps F \subseteq (\eps \lat \cap (K-\eps K_0)) + \eps K_0 \\
  &\subseteq (K-\eps K_0) + \eps K_0 = K + \eps(K_0-K_0) 
\end{split}
\end{align}

From the above inclusions, we get that 
\[
\vol_n(K) \leq \vol_n((\eps \lat \cap (K-\eps K_0)) + \eps F) \leq \vol_n(K+\eps(K_0-K_0)) \text{.}
\]

Since $F$ tiles with respect to $\lat$, we see that
\begin{align*}
\vol_n((\eps \lat \cap (K-\eps K_0)) + \eps F) &= |\eps \lat \cap (K-\eps K_0)| \vol_n(\eps F) \\
                                                                 &= \eps^n \det(\lat) |\eps \lat \cap (K-\eps K_0)|\text{, } \quad \text{ as needed.}
\end{align*}

For the furthermore, we assume that $K_0 \subseteq K-\vecc$ and that $K_0$ is symmetric. By symmetry of $K_0$, we have that
$\pm F \subseteq K_0 \subseteq K - \vecc$. Using this, we modify the inclusions in Equation \eqref{eq:vest-eq1} to
\begin{align}
\label{eq:vest-eq2}
\begin{split}
K &\subseteq (\eps \lat \cap (K - \eps F)) + \eps F \subseteq (\eps \lat \cap (K + \eps (K - \vect)) + \eps F \\
  &\subseteq (\eps \lat \cap ((1+\eps)K - \eps \vecc)) + \eps F \subseteq (\eps \lat \cap ((1+\eps)K - \eps \vecc)) + \eps(K-\vecc)  \\
  &\subseteq (1+\eps)K - \eps \vecc + \eps (K-\vecc) = (1+2\eps)K - 2\eps \vecc
\end{split}
\end{align}

From here, the same argument as above combined with the identity $\vol_n((1+2\eps)K) = (1+2\eps)^n \vol_n(K)$ completes the proof of Lemma.
\end{proof}

We now prove the main volume estimation result. We note that if the input body $K$ is symmetric, the following algorithm will be able to directly use
the thin covering lattice construction for symmetric bodies (Theorem~\ref{thm:thin-lat}) without passing through the construction of
Theorem~\ref{thm:asym-thin-covering}. We will use this fact within our algorithm for finding approximate KB points (Theorem~\ref{thm:approx-kb}).

\begin{proof}[Proof of Theorem~\ref{thm:vol-est} (Volume Estimation)]
Given a convex body $K \subseteq \R^n$, we wish to compute $V$ such that 
\[
\vol_n(K) \leq V \leq (1+\eps)^n \vol_n(K) \text{.}
\]
Compute the lattice $\Lambda$ with basis $B$ and point $\vecc \in K$ given by Theorem~\ref{thm:asym-thin-covering}. This requires $2^{O(n)}$ time and $\poly(n)$ space. 
Via enumeration, we now compute the quantity
\[
V = (\eps/2)^n \det(\Lambda) ~|(\eps/2)\Lambda \cap ((1+\eps/2)K-(\eps/2)\vecc)| \text{.}
\]
Since $K[\vecc] \subseteq K-\vecc$ and $\Lambda$ is $K[\vecc]$-covering, by Lemma \ref{lem:points-to-vol} we have that $V$ satisfies the desired guarantees. 

After rescaling, computing $V$ can be done by enumerating $\Lambda \cap ((1+2/\eps)K-\vecc)$. Since $K[\vecc]$ is
$2^{O(n)}$-SE with respect to $B$ by Lemmas \ref{lem:se-enum} and \ref{lem:se-cover-bounds} this enumeration complexity is bounded by
\begin{align*}
2^{O(n)} N((1+2/\eps)K,K[\vecc]) &\leq 2^{O(n)} \frac{\vol_n((1+2/\eps)K + K[\vecc])}{\vol_n(K[\vecc])} \\
                                    &\leq 2^{O(n)} (2+2/\eps)^n \frac{\vol_n(K)}{\vol_n(K[\vecc])} \\
                                    &\leq 2^{O(n)} (1+1/\eps)^n \text{.}
\end{align*}
Hence the total time complexity of the algorithm is bounded by $2^{O(n)} (1+1/\eps)^n$ and the space complexity is $\poly(n)$ as needed. 
\end{proof}

\subsection{Computing an Approximate Kovner-Besicovitch Point} 
\label{sec:approx-kb}

\begin{proof}[Proof of Theorem~\ref{thm:approx-kb} (Computing Kovner-Besicovitch points)] \hspace{1em} \\
Here the goal is to compute a point $\vecc \in K$ such that 
\[
\vol_n(K[\vecc])/\vol_n(K) \geq (1+\eps)^{-n} {\rm Sym}_{kb}(K) \text{.}
\]
By first applying deterministic ellipsoidal rounding to $K$ (Theorem~\ref{thm:gls-round}), we may assume that 
\[
B_2^n \subseteq K \subseteq (n+1)n^{1/2} K \text{.}
\]
We define the following sequence of bodies: $K_i = 2^i B_2^n \cap K$, for $0 \leq i \leq T$, \\
where $T = \ceil{\log_2 (n+1)n^{1/2}}$. By construction $K_0 = B_2^n$ and $K_T = K$. \\

\begin{algorithm}
\caption{Improve($A,~\vecx,~\alpha,~\eps$)}
\begin{algorithmic}[1]
\REQUIRE Convex body $A \subseteq \R^n$, point $\vecx \in A$ satisfying $\vol_n(A[\vecx]) \geq \alpha^n \vol_n(A)$, $\eps \leq 1/2$.
\ENSURE A point $\vecc \in A$ satisfying $\vol_n(A[\vecx]) \geq (1+\eps)^{-n} {\rm Sym}_{kb}(A)$.
\STATE $\eps_0 \leftarrow \eps/(6+3\eps)$, $J \leftarrow \floor{\log(1/\alpha)/\log(1/(1-\eps_0))}$.
\STATE $\vecx_0 \leftarrow \vecx$.
\FOR{$j \in 1 ~\TO~ J$}
\STATE Compute a covering $N$ of $1/2(A + \vecx_{j-1})$ by $(\eps_0/2) A[\vecx_{j-1}]$ (Theorem \ref{thm:thin-lat}). \\
       For each $\vecy \in N$, estimate the volume of $A[\vecy]$ to within $(1+\eps_0/(1-\eps_0))^n$ (Theorem \ref{thm:vol-est}). \\
       Set $\vecx_j$ to be the center in $N$ of maximum estimated volume.
\ENDFOR
\RETURN $\vecx_J$.
\end{algorithmic} 
\label{alg:improve}
\end{algorithm}

Using the improvement procedure (Algorithm~\ref{alg:improve}), the remainder of the algorithm is straightforward:

\begin{algorithmic}[1]
\STATE $\vecc_0 \leftarrow \veczero$.
\FOR{$i \in 1 ~\TO~T-1$}
  \STATE $\vecc_i \leftarrow {\rm Improve}(K_i,~\vecc_{i-1},~1/6,~1/2)$.
\ENDFOR
\RETURN ${\rm Improve}(K_T,~\vecc_{T-1},~1/6,~\eps)$.
\end{algorithmic}
\vspace{1em}
\noindent We first argue the correctness of the algorithm, and then continue with its runtime analysis.

\paragraph{{\bf Correctness}:} Assuming the correctness of Algorithm~\ref{alg:improve}, we show that the remainder of the algorithm is correct.
For the for loop on lines $2-3$, and line $4$, we claim that at each call ${\rm Improve}(K_i,~\vecc_{i-1},~1/6,~\dots)$, $\vecc_{i-1}$ has KB value
at least $(1/6)^n$ with respect to $K_i$, for $i \in [T]$. We prove this by induction on $i \in [T]$. Note that if $\vecc_{i-1}$ satisfies the condition, then
by the guarantess on ${\rm Improve}$, we have that
\[
\frac{\vol_n(K_i[\vecc_i])}{\vol_n(K_i)} \geq (1+1/2)^{-n} {\rm Sym}_{kb}(K_i) \geq (1+1/2)^{-n} 2^{-n} = 3^{-n}
\] 
From here, since $K_i \subseteq K_{i+1} \subseteq 2K_i$, we have that
\[
\frac{\vol_n(K_{i+1}[\vecc_i])}{\vol_n(K_{i+1})} \geq \frac{\vol_n(K_i[\vecc_i])}{\vol_n(K_{i+1})} 
                                                 \geq 3^{-n} \frac{\vol_n(K_i)}{\vol_n(K_{i+1})} \geq 3^{-n} ~ 2^{-n} = (1/6)^n \text{,}
\]
as needed. For the base case $i=1$, we note that since $\vecc_0 = \veczero$ and $K_0 = B_2^n$, $\veczero$ has KB value $1$ for $K_0$.
By the above analysis, we get that $\vecc_0$ has KB value at least $2^{-n} \geq (1/6)^n$ for $K_1$, as needed. 

Since on line $4$ we call ${\rm Improve}(K_T,~\vecc_{T-1},~1/6,~\eps)$ on a valid input and $K_T = K$, by the guarantees on ${\rm Improve}$,
the algorithm correctly outputs a $(1+\eps)^{-n}$ KB point for $K$ as needed. 
\vspace{1em}

\noindent We now show that Algorithm ${\rm Improve}$ is correct. Define 
\[
\nu(\vecx) = \left(\frac{\vol_n(A[\vecx])}{\vol_n(A)}\right)^{1/n}
\]
to be the normalized KB value of a point $\vecx \in A$. 

\begin{claim} $\nu$ is a concave function over $A$. \label{cl:val-improv} \end{claim}
\begin{proof}
Take $\vecx,\vecy \in K$ and $\alpha \in [0,1]$. By convexity of $A$, note that
\begin{align*}
\alpha A[\vecx] + (1-\alpha)A[\vecy] &= \alpha(A-\vecx) \cap (\vecx-A) + (1-\alpha)(A-\vecy) \cap (\vecy-A) \\
                                     &\subseteq \left(\alpha(A-\vecx) + (1-\alpha)(A-\vecy)\right) \cap \left(\alpha(\vecx-A) + (1-\alpha)(\vecy-A)\right) \\
                                     &= (A - (\alpha \vecx + (1-\alpha) \vecy)) \cap ((\alpha \vecx + (1-\alpha) \vecy) - A) = A[\alpha \vecx + (1-\alpha)\vecy]
\end{align*}
Using the above inclusion, followed by the Brunn-Minkowski inequality, we get that
\begin{align*}
\vol_n(A[\alpha \vecx + (1-\alpha)\vecy])^{1/n} &\geq \vol_n(\alpha A[\vecx] + (1-\alpha)A[\vecy])^{1/n} \\
                                                &\geq \alpha \vol_n(A[\vecx])^{1/n} + (1-\alpha) \vol_n(A[\vecy])^{1/n} \text{.}
\end{align*}
The claim follows by dividing through by $\vol_n(A)^{1/n}$.
\end{proof}

\noindent Let $\vecx^*$ denote the center of maximum KB value, i.e.~ $\vecx^* = \argmax_{\vecx \in A} \nu(\vecx)$, and let $\gamma = \nu(\vecx^*)$. 
Note that for correctness, we need simply show that at the last iteration $J$, $\nu(\vecx_J) \geq \gamma/(1+\eps)$.
The following claim tracks the progress in $\nu$.

\begin{claim} For $i \geq 1$, $\nu(\vecx_i) \geq 1/2(\gamma + \nu(\vecx_{i-1}))(1-\eps_0)^2$. \label{cl:up-value} \end{claim}
\begin{proof}
By translating $A$, we may assume that $\vecx_{i-1} = \veczero$. Let $\vecz = 1/2 \vecx^*$. By construction $\vecz \in 1/2A$, 
hence by the properties of the net $N$, there exists $\vecy \in N$ such that $\vecv = \vecy-\vecz$ satisfies $\|\pm \vecv\|_{A[\veczero]} \leq \eps_0/2$.
By the triangle inequality, note that
\[
\|\vecz + 1/\eps_0 \vecv\|_A \leq \|\vecz\|_A + 1/\eps_0\|\vecv\|_A \leq 1/2 + 1/\eps_0\|\vecv\|_{A[\veczero]} \leq 1/2 + 1/\eps_0(\eps_0/2) \leq 1 \text{.}
\]
Hence $\vecz + 1/\eps_0 \vecv \in A$. Since $\vecy = \vecz + \vecv = (1-\eps_0)\vecz + \eps_0(\vecz + 1/\eps_0 \vecv)$, by concavity of $\nu$ over $A$
\begin{align}
\label{eq:net-good}
\begin{split}
\nu(\vecy) &\geq (1-\eps_0) \nu(\vecz) + \eps_0 \nu(\vecz + 1/\eps_0 \vecv) \geq (1-\eps_0) \nu(\vecz) \\ 
           &= (1-\eps_0) \nu((1/2) \veczero + (1/2) \vecx^*) \geq (1-\eps_0)((1/2) \nu(\veczero) + (1/2)\nu(\vecx^*)) \\ 
           &= 1/2(\nu(\vecx_{i-1}) + \gamma)(1-\eps_0)
\end{split}
\end{align}
For each $\vecy \in N$, we note that volume estimation algorithm computes a number $V_\vecy$ such that 
\[
\vol_n(A[\vecy]) \leq V_\vecy \leq (1+\eps_0/(1-\eps_0))^n \vol_n(A[\vecy]) = 1/(1-\eps_0)^n \vol_n(A[\vecy]) \text{.}
\]
By Equation \eqref{eq:net-good}, this implies that for the chosen $\vecx_i$, we must have 
\[
V_{\vecx_i} \geq \vol_n(A)\left(1/2(\nu(\vecx_{i-1}) + \gamma)(1-\eps_0)\right)^n. 
\]
By approximation the guarantee, this implies that $\nu(\vecx_i) \geq 1/2(\nu(\vecx_{i-1}) + \gamma)(1-\eps_0)^2$, as needed.
\end{proof}

The following claim completes the proof of correctness:

\begin{claim} At the last iteration $J = \floor{\log(1/\alpha)/\log(1/(1-\eps_0))}$, $\nu(\vecx_J) \geq \gamma/(1+\eps).$ \label{cl:final-conv} \end{claim}
\begin{proof}
Let $a_0 = \alpha$, and let $a_i = 1/2(a_{i-1} + \gamma)(1-\eps_0)^2$ for $i \geq 1$. Since the function $a \rightarrow 1/2(a + \gamma)(1-\eps_0)^2$
is monotone in $a$, by Claim \ref{cl:up-value} we have that $\nu(\vecx_i) \geq a_i$ for all $i$.  It therefore suffices to prove that $a_J \geq
\gamma/(1+\eps)$. We first note that $\eps_0 = \eps/(6+3\eps)$ is set to satisfy the equation $(1-3\eps_0)/(1+3\eps_0) = 1/(1+\eps)$.  If $a_{i-1}
\leq \gamma/(1+\eps) = \gamma(1-3\eps_0)/(1+3\eps_0)$, note that
\begin{align*}
a_i(1-\eps_0) &= 1/2(a_{i-1} + \gamma)(1-\eps_0)^3 \geq 1/2(a_{i-1} + \gamma)(1-3\eps_0) \\
              &= 1/2(a_{i-1}(1-3\eps_0) + \gamma (1-3\eps_0))  \geq 1/2(a_{i-1}(1-3\eps_0) + a_{i-1} (1+3\eps_0)) = a_{i-1} \text{.}
\end{align*} 
In particular, we get $a_i \geq a_{i-1}/(1-\eps_0)$. Furthermore, if $a_{i-1} \geq \gamma/(1+\eps)$ by monotonicity $a_i \geq \gamma/(1+\eps)$.
Therefore, we need only show that the $a_i$ goes above $\gamma/(1+\eps)$ at some time $i \leq J$. Let $t$ be the first step where $a_t \geq
\gamma/(1+\eps)$. By the above relations, we must have that
\[
1 \geq \nu(x_t) \geq a_t \geq a_{t-1}/(1-\eps_0) \geq a_0/(1-\eps_0)^t = \alpha/(1-\eps_0)^t \text{.}
\] 
Solving for $t$, we get that $t \leq \log(1/\alpha)/\log(1/(1-\eps_0))$, and hence $t \leq J$ as needed.
\end{proof}

\paragraph{{\bf Runtime Analysis}:} We first apply ellipsoidal rounding to $K$ (Theorem~\ref{thm:gls-round}), this can be done in polynomial time.
Next, we run the ${\rm Improve}$ procedure $O(\log n)$ times, so it suffices to bound the runtime of one call. Since without loss of generality we can
assume $\eps \leq 1/2$, it is clear that the last call to procedure ${\rm Improve}$, that is ${\rm Improve}(K_T,\vecc_{T-1},1/6,\eps)$, dominates the
complexity of the algorithm.

On the last call to Improve, we have $A=K_T$, $\alpha=1/6$, $\eps \leq 1/2$, and $\eps_0 = \eps/(6+3\eps_0) \geq \eps/8$. We execute the main loop  
\[
\log(1/\alpha)/\log(1/(1-\eps_0)) \leq \log(1/\alpha)/\log(1+\eps_0) \leq \log(1/\alpha)/\log(1+\eps/8) = O(1/\eps) \text{ times.}
\]
Let $\gamma = {\rm Sym}_{kb}(A)^{1/n}$. Note that $\gamma/(1+\eps) \geq (1/2)(2/3) = 1/3 \geq \alpha$. Hence
by Claim \ref{cl:final-conv} at each iteration of the for loop we have that $\vol_n(A[\vecx_{j-1}])/\vol_n(A) \geq 6^{-n}$.

At iteration $j$, we first compute an covering $N$ of $1/2(A+\vecx_{j-1})$ by $(\eps_0/2) A[\vecx_{j-1}]$ using Theorem~\ref{thm:thin-lat}.
Since $\eps_0/2 \geq \eps/16$, this takes time at most
\begin{align*}
2^{O(n)}N(1/2 A, (\eps/16)A[\vecx_{j-1}]) &= 2^{O(n)} \frac{\vol_n(1/2 A + (\eps/16)A[\vecx_{j-1}])}{\vol_n((\eps/16)A[\vecx_{j-1}]} \\
                                      &= 2^{O(n)}\left(\frac{1/2 + \eps/16}{\eps/16}\right)^n \frac{\vol_n(A)}{\vol_n(A[\vecx_{j-1}])} \\
                                      &= 2^{O(n)}(1 + 8/\eps)^n 6^n = 2^{O(n)}(1 + 1/\eps)^n
\end{align*}
and $\poly(n)$ space. For each $\vecy \in N$, we compute a number $V_y$ satisfying 
\[
\vol_n(A[\vecy]) \leq V_y \leq (1+\eps_0/(1-\eps_0))^n \vol_n(A[\vecy])  
\]
where $\eps_0/(1-\eps_0) \geq 2 \eps_0 \geq \eps/4$. Since $A[\vecy]$ is symmetric, we note that this can be done using Theorem \ref{thm:vol-est},
using only the thin lattice construction for symmetric bodies (Theorem ~\ref{thm:thin-lat}). Hence this can be done in $2^{O(n)}(1+4/\eps)^n =
2^{O(n)}(1+1/\eps)^n$ time using $\poly(n)$ space. Putting it all together, the for loop can be executed in $2^{O(n)}(1/\eps)(1 + 1/\eps)^{2n} =
2^{O(n)}(1+1/\eps)^{2n+1}$ time using $\poly(n)$ space. The desired complexity bound for the algorithm follows.
\end{proof}

\paragraph{{\bf Acknowledgments.}} The author would like to thank Santosh Vempala and Oded Regev for useful conversations related to this paper.

\bibliographystyle{alpha}
\bibliography{bib/lattices,bib/acg,bib/cg,bib/reference-books}

\newcommand{\etalchar}[1]{$^{#1}$}
\begin{thebibliography}{ASL{\etalchar{+}}13}

\bibitem[AI06]{conf/focs/IndykA06}
Alexandr Andoni and Piotr Indyk.
\newblock Near-optimal hashing algorithms for approximate nearest neighbor in
  high dimensions.
\newblock In {\em Proceedings of the 47th Annual IEEE Symposium on Foundations
  of Computer Science (FOCS)}, pages 459--468, 2006.

\bibitem[Ajt04]{ajtai04:_gener_hard_instan_lattic_probl}
Mikl{\'o}s Ajtai.
\newblock Generating hard instances of lattice problems.
\newblock {\em Quaderni di Matematica}, 13:1--32, 2004.
\newblock Preliminary version in STOC 1996.

\bibitem[ASL{\etalchar{+}}13]{conf/stoc/ASLV13}
N.~Alon, A.~Schraibman, T.~Lee, , and S.~Vempala.
\newblock The approximate rank of a matrix and its algorithmic applications.
\newblock In {\em STOC}, 2013.

\bibitem[BGJ13]{eprint/svp/BGJ13}
A.~Becker, N.~Gama, and A.~Joux.
\newblock Solving shortest and closest vector problems: The decomposition
  approach.
\newblock Cryptology Eprint. Report 2013/685, 2013.

\bibitem[But72]{Butler1972}
G.J. Butler.
\newblock Simultaneous packing and covering in euclidean space.
\newblock {\em Proceedings of the London Mathematical Society}, 25(3):721--735,
  1972.

\bibitem[Dad12]{thesis/D12}
Daniel Dadush.
\newblock {\em Integer Programming, Lattice Algorithms, and Deterministic
  Volume Estimation}.
\newblock PhD thesis, Georgia Institute of Technology, 2012.

\bibitem[Dad13]{journal/alg/D13}
Daniel Dadush.
\newblock A randomized sieving algorithm for approximate integer programming.
\newblock {\em Algorithmica}, pages 1--37, September 2013.
\newblock Preliminary version in LATIN 2012.

\bibitem[DFK91]{DFK}
M.E. Dyer, A.M. Frieze, and R.~Kannan.
\newblock A random polynomial-time algorithm for approximating the volume of
  convex bodies.
\newblock {\em J. ACM}, 38(1):1--17, 1991.

\bibitem[DK13]{conf/soda/cvp/DK13}
Daniel Dadush and Gabor Kun.
\newblock Lattice sparsification and the approximate closest vector problem.
\newblock In {\em SODA}, 2013.

\bibitem[DPV11]{conf/focs/svp/DPV11}
Daniel Dadush, Chris Peikert, and Santosh Vempala.
\newblock Enumerative lattice algorithms in any norm via m-ellipsoid coverings.
\newblock In {\em FOCS}, 2011.

\bibitem[DV13]{journal/pnas/DV13}
Daniel Dadush and Santosh~S. Vempala.
\newblock Near-optimal deterministic algorithms for volume computation via
  m-ellipsoids.
\newblock {\em Proceedings of the National Academy of Sciences}, 2013.

\bibitem[ELZ05]{journal/inf-theory/ELZ05}
Uri Erez, Simon Litsyn, and Ram Zamir.
\newblock Lattices which are good for (almost) everything.
\newblock {\em Information Theory, IEEE Transactions on}, 51(10):3401--3416,
  2005.

\bibitem[FB86]{BF87}
Z.~F{\"u}redi and I.~B{\'a}r{\'a}ny.
\newblock Computing the volume is difficult.
\newblock In {\em STOC '86: Proceedings of the eighteenth annual ACM symposium
  on Theory of computing}, pages 442--447, New York, NY, USA, 1986. ACM.

\bibitem[FB88]{BF88}
Z.~F{\"u}redi and I.~B{\'a}r{\'a}ny.
\newblock Approximation of the sphere by polytopes having few vertices.
\newblock {\em Proceedings of the AMS}, 102(3), 1988.

\bibitem[GLS88]{GLS}
M.~Gr{\"o}tschel, L.~Lov{\'a}sz, and A.~Schrijver.
\newblock {\em Geometric Algorithms and Combinatorial Optimization}.
\newblock Springer-Verlag, 1988.

\bibitem[GM03]{goldstein03:_equid_of_hecke_point}
Daniel Goldstein and Andrew Mayer.
\newblock On the equidistribution of hecke points.
\newblock {\em Forum Mathematicum}, 15(2):165--189, 2003.

\bibitem[GMR05]{journal/cc/GuMiRe05}
Venkatesan Guruswami, Daniele Micciancio, and Oded Regev.
\newblock The complexity of the covering radius problem.
\newblock {\em Computational Complexity}, 14(2):90--121, 2005.
\newblock Preliminary version in CCC 2004.

\bibitem[Gr{\"u}61]{Grunbaum61}
B.~Gr{\"u}nbaum.
\newblock Measures of symmetry for convex sets.
\newblock In {\em Proceedings of the 7th Symposium in Pure Mathematics of the
  American Mathematical Society, Symposium on Convexity}, pages 233--270, 1961.

\bibitem[HR06]{DBLP:conf/coco/HavivR06}
Ishay Haviv and Oded Regev.
\newblock Hardness of the covering radius problem on lattices.
\newblock In {\em IEEE Conference on Computational Complexity}, pages 145--158,
  2006.

\bibitem[HS07]{DBLP:conf/crypto/HanrotStehle07}
Guillaume Hanrot and Damien Stehl\'{e}.
\newblock Improved analysis of kannan's shortest lattice vector algorithm.
\newblock In {\em Proceedings of the 27th annual international cryptology
  conference on Advances in cryptology}, CRYPTO'07, pages 170--186, Berlin,
  Heidelberg, 2007. Springer-Verlag.

\bibitem[Kho05]{DBLP:journals/jacm/Khot05}
Subhash Khot.
\newblock Hardness of approximating the shortest vector problem in lattices.
\newblock {\em J. ACM}, 52(5):789--808, 2005.
\newblock Preliminary version in FOCS 2004.

\bibitem[Mil86]{M86}
V.D. Milman.
\newblock Inegalites de brunn-minkowski inverse et applications at la theorie
  locales des espaces normes.
\newblock {\em C. R. Acad. Sci. Paris}, 302(1):25--28, 1986.

\bibitem[MP00]{MP00}
V.D. Milman and A.~Pajor.
\newblock Entropy and asymptotic geometry of non-symmetric convex bodies.
\newblock {\em Advances in Mathematics}, 152(2):314 -- 335, 2000.

\bibitem[MR07]{DBLP:journals/siamcomp/MicciancioR07}
Daniele Micciancio and Oded Regev.
\newblock Worst-case to average-case reductions based on {Gaussian} measures.
\newblock {\em SIAM J. Comput.}, 37(1):267--302, 2007.
\newblock Preliminary version in FOCS 2004.

\bibitem[MV13]{journal/siamjc/MV13}
D.~Micciancio and P.~Voulgaris.
\newblock A deterministic single exponential time algorithm for most lattice
  problems based on voronoi cell computations.
\newblock {\em SIAM Journal on Computing}, 42(3):1364--1391, 2013.
\newblock Preliminary version in STOC 2010.

\bibitem[Pei09]{DBLP:conf/stoc/Peikert09}
Chris Peikert.
\newblock Public-key cryptosystems from the worst-case shortest vector problem.
\newblock In {\em STOC}, pages 333--342, 2009.

\bibitem[Reg09]{DBLP:journals/jacm/Regev09}
Oded Regev.
\newblock On lattices, learning with errors, random linear codes, and
  cryptography.
\newblock {\em J. ACM}, 56(6):1--40, 2009.
\newblock Preliminary version in STOC 2005.

\bibitem[Rog50]{Rogers1950}
C.~A. Rogers.
\newblock A note on coverings and packings.
\newblock {\em Journal of the London Mathematical Society}, s1-25(4):327--331,
  1950.

\bibitem[Rog58]{Rogers1958}
C.~A. Rogers.
\newblock Lattice coverings of space: The minkowski-hlawka theorem.
\newblock {\em Proceedings of the London Mathematical Society},
  s3-8(3):447--465, 1958.

\bibitem[Rog59]{Rogers1959}
C.A. Rogers.
\newblock Lattice coverings of space.
\newblock {\em Mathematika}, 6:33--39, 6 1959.

\bibitem[RS57]{RS57}
C.A. Rogers and G.C. Shephard.
\newblock The difference body of a convex body.
\newblock {\em Arch. Math.}, 8:220--233, 1957.

\bibitem[RZ97]{RogersZong97}
C.~A. Rogers and C.~Zong.
\newblock Covering convex bodies by translates of convex bodies.
\newblock {\em Mathematika}, 44:215--218, 6 1997.

\bibitem[Vem13]{error/V13}
S.~Vempala.
\newblock Personal communication, 2013.

\bibitem[YN76]{YN76}
D.~B. Yudin and A.~S. Nemirovski.
\newblock Evaluation of the information complexity of mathematical programming
  problems (in russian).
\newblock {\em Ekonomika i Matematicheskie Metody}, 13(2):3--45, 1976.

\end{thebibliography}

\end{document}